\documentclass[showpacs,nofootinbib,showkeys,eqsecnum,prd,aps,preprint]{revtex4-1}

\usepackage[english]{babel}

\usepackage{amssymb}
\usepackage{amsmath}
\usepackage{amsfonts}
\usepackage{hyperref}

\newtheorem{lemma}{Lemma}
\newtheorem{corollary}{Corollary}
\newtheorem{theorem}{Theorem}

\newenvironment{proof}[1][Proof]
{\textbf{#1} }{\ \rule{0.5em}{0.5em}}

\begin{document}
	
	\title{Noncommutative integration of the Dirac equation in homogeneous spaces}
	
	\author{A. I. Breev }
	\email{breev@mail.tsu.ru}
	\affiliation{Department of Theoretical Physics, Tomsk State University Novosobornaya
		Sq. 1, Tomsk, Russia, 634050}
	
	\author{A. V. Shapovalov }
	\email{shpv@phys.tsu.ru}
	\affiliation{Department of Theoretical Physics, Tomsk State University Novosobornaya
		Sq. 1, Tomsk, Russia, 634050}
	\affiliation{Tomsk Polytechnic University,  Lenin ave., 30, Tomsk, Russia, 634034}
	
	\begin{abstract}
		We develop a noncommutative integration method for the  Dirac equation in homogeneous spaces. The Dirac equation with an invariant metric is shown to be equivalent to a system of equations on a Lie group of transformations of a homogeneous space. This allows us to effectively apply the noncommutative integration method of linear partial differential equations on Lie groups. This method differs from the well-known method of separation of variables and to some extent can often supplement it.
		The general structure of the method developed is illustrated with an  example of a homogeneous space which does not admit separation of variables in the Dirac equation. However, the basis of exact solutions to the Dirac equation is constructed explicitly by the noncommutative integration method. Also, we construct a complete set of new exact solutions to the Dirac equation in the three-dimensional de Sitter space-time $\mathrm{AdS_{3}}$ using the method developed. The solutions obtained are  found in terms of elementary functions, which is characteristic of the noncommutative integration method.
	\end{abstract}
	
	

\maketitle
		
\section{Introduction\label{S1}}

Exact solutions of the relativistic wave equations in strong gravitational and electromagnetic fields are the basis for studying quantum effects in the framework of quantum field theory in curved space-time (see, e.g. \cite{Birrel,Grib,Flh,Fvz,Wr,ksno}). A construction of the complete set of exact solutions to these equations in many cases is associated with the presence of integrals of motion. For example, to separate the variables in a wave equation, it is necessary to have $\dim{M}-1$ commuting integrals, where $M$ is the space of independent variables. In this paper, by integrability of the wave equation we mean an explicit possibility of reducing the original equation to a system of ordinary differential equations, the solution of which provides a complete set  of solutions to the original wave  equation.

The best-known technique for such a reduction is based on the method of separation of variables (SoV) (various aspects of the SoV method can be found, e.g., in \cite{Kalnins,Kalnins2,Miller}). There is a broad scope of research dealing with separation of variables in relativistic quantum wave equations, mainly for the Klein-Gordon and Dirac equations, and with classification of external fields admitting SoV in these equations (see, e.g., \cite{Bagrov} and references therein). This motivates the development of methods for the exact integration of wave equations other than SoV that can give some new possibilities in relativistic quantum theory.

In this regard, we focus on homogeneous spaces as geometric objects with high symmetry.
We also note that most of the physically interesting problems and effects are associated with gravitational fields possessing symmetries. Mathematically, these symmetries indicate the pre\-sen\-ce of various groups of transformations that leave invariant the gravitational field. Represen\-ting the space-time as a homogeneous space with a group-invariant metric, we can consider a large class of gravitational fields and cosmological models \cite{Stephani,Ryan} with rich symmetries, and the corresponding relativistic equations in these fields have integrals of motion.

We note that the relativistic wave equations on a homogeneous space may not allow separation of variables. The matter is that in accordance with the theorem of Refs. \cite{ShVN1,ShVN2}, for the separation of variables in the wave equation in an appropriate coordinate system the equation should admit a complete set of mutually commuting symmetry operators (integrals of motion, details can be found in \cite{ShVN1,ShVN2}, see also \cite{Ob2020}). Therefore, the problem arises of constructing exact solutions to the wave equation in the case when it has symmetry operators, but they do not form a complete set and separation of variables can not be carried out.
We consider the noncommutative integration method (NCIM) based on noncommutative algebras of symmetry operators admitted by the equation \cite{SpSh1,Sh2003,ShSO3,ShBr11,ShBr14}. This method can be thought as a generalization of the method of SoV. A reduction of the wave equation to a system of ODEs according to the NCIM (we use the term noncommutative reduction) can be carried out in a way  that is substantially different from the method of separation of variables.

We note that the method of noncommutative integration has shown its effectiveness in constructing bases of exact solutions to the Klein-Gordon and Dirac equations in some spaces with invariance groups.

For instance, the NCIM was applied to the Klein-Gordon equation in homogeneous spaces with an invariant metric in \cite{ShBr11,ShBr14}. The polarization vacuum effect of a scalar field in a homogeneous space was studied using NCIM in \cite{ShBr11,ShBr14,BrKz}.

The noncommutative reduction of the Dirac equation to a system of ordinary differential equations in the Riemannian and pseudo-Riemannian spaces with a nontrivial group of motions was considered in \cite{fed01,sh02,var03,var04,kls01,kls02}. In Refs. \cite{kls03,kls04} the NCIM was applied to the Dirac equation in the four-dimensional flat space and in the de Sitter space. The Dirac equation on Lie groups that can be a special case of homogeneous spaces with a trivial isotropy subgroup,  was explored in terms of the NCIM in Refs. \cite{brr01,brr02}.

It may also be worth noting  that the application of the NCIM to the Dirac equation can give a new class of its exact solutions, different from the solutions obtained by SoV. In cases where the Dirac equation does not admit separation of variables, the NCIM provides an uncontested option for constructing complete sets of solutions. The physical meaning of the solutions obtained by this method depends on the specifics of the problem being solved and requires special research in each case.

In the present work we consider non-commutative symmetries of the Dirac equation in homogeneous spaces. We also develop the method of noncommutative integration of the Dirac equation in homogeneous spaces. Using the group-theoretic approach, we reduce the Dirac equation on the homogeneous space to such a system of equations on the transformation group that lets us to apply the noncommutative reduction and construct exact solutions of  the Dirac equation.
In this paper, for the first time, we explicitly take into account the identities for generators of the transformation group in the problem of noncommutative reduction for the Dirac equation.

The work is organized as follows. In Section \ref{invariant-mertic} we briefly introduce basic concepts and notations from the theory of homogeneous spaces \cite{DiffG,Kobyashi,AdAr}, to be used later.

A construction of invariant differential operator with matrix coefficients on a homogeneous space is introduced in Section \ref{invariant-matrix-operator} following Refs.\cite{KyrnScalar,BarProlong}. Also in this section, we show the connection between generators of the representation of a Lie group on a homogeneous space and the other  representation induced by representation of  a subgroup, whose action on a homogeneous space has a stationary point.

In the next Section \ref{sec:lambda}, we introduce a special irreducible representation of the Lie algebra of the Lie group of transformations of a homogeneous space using the Kirillov orbit method \cite{Kirr}, that is necessary for noncommutative reduction.

In Section \ref{dirac-eq-in homogen-space} we present the Dirac equation in a homogeneous space with an invariant metric in terms of an invariant matrix operator of the first order. The spinor connection and symmetry operators of the Dirac equation are shown to define isotropy representation in a spinor space. Generators of the spinor representation are found explicitly.

We also introduce a system of differential equations on the Lie group of transformations of a homogeneous space, which is equivalent to the original Dirac equation in a homogeneous space.

Then, in Section \ref{sec:non},  we present  a noncommutative reduction of the Dirac equation on a homogeneous space, using the irreducible $\lambda$-representation introduced in section  \ref{sec:lambda} and  functional relations between  symmetry operators (identities) for the Dirac equation.

In Section \ref{sec:nosh}, we consider a homogeneous space with an invariant metric that does not admit separation of variables for the Klein--Gordon and Dirac equations. In this case a complete set of exact solutions of the Dirac equation is constructed using the noncommutative reduction (Section \ref{sec:non}).

The next Section \ref{sec:ads3} is devoted to the Dirac equation in the $(2+1)$ anti-de Sitter $\mathrm{AdS_{3}}$ - dimensional space. In this homogeneous space there are identities between the generators of the representation of the group $SO(1,3)$ that is taken into account when the noncommutative reduction is applied. The Dirac equation admits separation of variables in the $\mathrm{AdS_{3}}$ space, but the separable solutions are expressed through the special functions and have a complex form. The NCIM, being applied to this problem, results in the other complete set of exact solutions to the Dirac equation in the $\mathrm{AdS_{3}}$ space and these solutions are presented in terms of elementary functions.

In Section \ref{sec:Conclusion-remarks} we give our conclusion remarks.

\section{Invariant metric on a homogeneous space}\label{invariant-mertic}

This section introduces some basic concepts and notations of the  homogeneous space theory with an invariant metric.

Let $G$ be a simply connected real Lie group with a Lie algebra $\mathfrak{g}$, $M$ be a homogeneous space with right action of the group $G$, $(x,g)\rightarrow R_{g}x=xg\in M$ for $x\in M$, $g\in G$. For any $x\in M$ there exists an isotropic subgroup $H_{x}\in G$. Denote by $H=H_{x_{0}}$ a closed stabilizer of a point $x_{0}\in M$, and let $\mathfrak{h}$ be a Lie algebra of $H$. The homogeneous space $M$ is diffeomorphic to a quotient manifold $G/H$ of right cosets $[Hg]$ of the Lie group $G$ by $H$.

A transformation group  $G$ can be regarded as a principal bundle $(G,\pi,M,H)$ with a structure group $H$, a base $M$, and a canonical projection $\pi:G\rightarrow M$, $\pi(e)=x_{0}$, where $e$ is the identity element of  $G$. An arbitrary point $g\in G$  can be represented uniquely as  $g=h\,s(x)$, where $x=x_{0}g=\pi(g)\in M$, $h\in H$, and $s:M\rightarrow G$ is a local and smooth section of $G$, $\pi\circ s=\mathrm{id}$.

Differential of the canonical projection $\pi_{*}:T_{g}G\rightarrow T_{\pi(g)}M$ is a surjective map that allows any tangent vector $\tau\in T_{x}M$ on a homogeneous space to be represented as $\pi_{*} \zeta$,  where $\zeta\in T_{g} G$ is a tangent vector on $G$.

In turn, a linear space of the Lie algebra $\mathfrak{g}\simeq T_{e}G$ is decomposed into a direct sum of subspaces $\mathfrak{g}=\mathfrak{h}\oplus\mathfrak{m}$, where $\mathfrak{m}=\pi_{*}(\mathfrak{g})\simeq T_{x_{0}}M$ is a complement to $\mathfrak{h}=\mathbf{ker\ }\pi_{*}(e)$, i.e. $X= X_{\mathfrak {m}} + X_{\mathfrak {h}}$ holds  for any $X\in \mathfrak {g}$, where $ X_{\mathfrak {m}} = \pi_{*} X $, $ X_{\mathfrak {h}} \in\mathfrak {h}$.

We introduce an invariant metric on the homogeneous space $M$. Let $\langle \cdot, \cdot \rangle_{\mathfrak {m}}$ be a non-degenerate $Ad(H)$ - invariant
scalar product on the subspace $\mathfrak {m}$,
\begin{gather}
\langle[X,Y]_{\mathfrak{m}},Z_{\mathfrak{m}}\rangle_{\mathfrak{m}}+\langle Y_{\mathfrak{m}},[X,Z]_{\mathfrak{m}}\rangle_{\mathfrak{m}}=0,\quad
 X\in\mathfrak{h}, \quad Y,Z\in\mathfrak{g}.\label{Mproduct}
\end{gather}

By action of a Lie group $G$ with right shifts on the homogeneous space $M$, we define the inner product throughout the space $M$ as
\begin{gather}
	\langle u,v\rangle_{x}=\langle(R_{g^{-1}})_{*}u,(R_{g^{-1}})_{*}v\rangle_{\mathfrak{m}}, \quad
 u,v\in T_{x}M,\quad x=\pi(g).\label{Minv}
 \end{gather}

The $Ad(H)$-invariance (\ref{Mproduct}) is necessary and sufficient for the inner product (\ref{Minv}) to be invariant with respect to the action of $G$ on $M$. The inner product defines an \textit{invariant} metric on the homogeneous space $M$ \cite{Kobyashi}.

On the principal bundle $(G,\pi,M,H)$, we introduce local coordinates $(x^{i},h^{\alpha})$, ($i=1$, $\dots$, $\mathrm{dim}\,M$, $\alpha=1,\dots,\mathrm{dim}\,H$) of the direct product $\text{U\ensuremath{\times H}}$, where $\{x^{i}\}$ are local coordinates in a domain $U\subset M$ of a trivialization covering  $x_{0}\in U$ with local coordinates $\{h^{\alpha}\}$ in the subgroup $H$, $e_{H}=\{0\}$. So the local coordinates of an element $g=h\,s(x)$ from some neighborhood of the identity $G_{e}$  can be represented as  $g^{i}=x^{i}$, $g^{\alpha}=h^{\alpha}$. We choose a section $s:M\rightarrow G$ so that equalities $s^{a}(x)=x^{a}$ and $s^{\alpha}(x)=0$ hold over the  domain  $U$.

The tangent vectors
\[
e_{a} = \left. s_{*} \partial_{x^{a}} \right |_{x = x_{0}}, \quad e_{\alpha} = \left. \partial_{h^{\alpha}} \right |_{h = e_{H}},
\]
form a basis $\{e_{A}\} = \{e_{a}\} \cup \{e_{\alpha} \}$, ($ A = 1, \dots, \mathrm{dim}\, \mathfrak{g}$) of the Lie algebra $\mathfrak {g} \simeq T_{x_{0}} M \oplus T_{e_{H}} H$, where $\{e_{\alpha}\} $ is a basis of the Lie algebra $\mathfrak {h \simeq} T_{e_{H}} H$, and $\{e_{a}\} $ is a basis of the linear space $\mathfrak {m \simeq} T_{x_{0}} M$.
In the conjugate space $\mathfrak {g}^{*} \simeq T_{e}^{*} G$, we introduce the dual basis $\{e^{A}\} $ by the condition $\langle e^{A}, e_{B}\rangle =\delta_{B}^{A}$ where $\delta_{B}^{A}$ is the Kronecker symbol. The right-invariant basis vector fields $\eta_{A}(g) =-(R_{g})_{*} e_{A} $ and the corresponding right-invariant 1-forms $\sigma^{A} (g) = - (R_{g})^{*} e^{A} $ in local coordinates $(x, h)$ have the form
\begin{gather*}
\eta_{a}(x,h) = \eta_{a}^{i}(x,h)\partial_{x^{i}}+\eta_{a}^{\alpha}(x,h)\partial_{h^{\alpha}},\quad
\eta_{\alpha}(x,h)=\eta_{\alpha}^{\beta}(h)\partial_{h^{\beta}},\\
\sigma^{a}(x,h)  =\sigma_{i}^{a}(x,h)dx^{i},\quad
\sigma^{\alpha}(x,h)=\sigma_{i}^{\alpha}(x,h)dx^{i}+\sigma_{\beta}^{\alpha}(h)dh^{\beta},\quad
\alpha,\beta=1,\dots,H.
\end{gather*}
Here $(R_{g})_{*}$, $(R_{g})^{*}$ are differentials of the right shifts $R_{g}(g')=gg'$ on the Lie group $G$.

The right-invariant vector fields $\eta_{A}$ satisfy the commutation relations $ [\eta_{A}, \eta_{B}] = C_{AB}^{C} \eta_{C}$, while the right-invariant 1-forms $\sigma^{A}$ satisfy the Maurer-Cartan relations, $ -2 \, d\sigma^{C} = C_{AB}^{C}\sigma^{A} \wedge \sigma^{B}$.
Here $ C_{AB}^{A} = [e_{A}, e_{B}]^{C}$ are the structure constants of $\mathfrak {g} $ with indices  $A, B, C = 1, \dots, \dim G $.

We consider an invariant metric on the homogeneous space $M$ using a coordinate system in a domain  $U$  of   trivialization. A symmetric non-degenerate square matrix $G_{ab}=\langle e_{a},e_{b}\rangle_{\mathfrak{m}}$
in a subspace $\mathfrak{m\subset g}$ satisfies the $Ad(H)$-condition
\begin{gather}\label{conds_AdH}
G_{ab}C_{c\alpha}^{a}+G_{ac}C_{b\alpha}^{a}=0, \\ \nonumber
a,b=1,\dots,\dim M,\quad \alpha=1,\dots,\dim H.
\end{gather}
The invariant metric tensor in local coordinates $\{x^{i}\}$
is written as \cite{magSigma2015}:
\begin{gather}\label{gij_loc}
g_{ij}(x)=G_{ab}\sigma_{i}^{a}(x,e_{H})\sigma_{j}^{b}(x,e_{H}),\quad i,j=1,\dots,\dim M. 
\end{gather}
The contravariant components of the metric tensor are
\[
g^{ij}(x)=G^{ab}\eta_{a}^{i}(x,e_{H})\eta_{b}^{j}(x,e_{H}),\quad G^{ab}=(G_{ab})^{-1}.
\]

In what follows we will need the Christoffel symbols of the Levi-Civita connection
with respect to a $G$-invariant metric $g_{M}$ given by \cite{ShBr11,Kobyashi}
\begin{align}\nonumber
\Gamma_{jk}^{i}(x)= & \Gamma_{bc}^{a}\sigma_{j}^{b}(x,e_{H})\sigma_{k}^{c}(x,e_{H})\eta_{a}^{i}(x,e_{H})-\\ \nonumber
 -& \sigma_{j}^{b}(x,e_{H})\eta_{b,k}^{i}(x,e_{H})- \\
 -&C_{b\alpha}^{a}\sigma_{j}^{b}(x,e_{H})\sigma_{k}^{\alpha}(x,e_{H})\eta_{a}^{i}(x,e_{H}).\label{gamma_P_IJK}
\end{align}
Here $i,j,k=1,\dots,\dim M$, and  $\Gamma_{bc}^{a}$ are determined by  $G^{ab}$ of the quadratic form $\mathbf{G}$ and the structure constants of the Lie algebra $\mathfrak{g}$,
\begin{equation}
\Gamma_{bc}^{a}=-\frac{1}{2}C_{bc}^{a}-\frac{1}{2}G^{ad}\left[G_{ec}C_{bd}^{e}+G_{eb}C_{cd}^{e}\right].\label{gamma_P}
\end{equation}
Thus, in a homogeneous space with invariant metric, the Levi-Civita connection is defined by algebraic properties of the homogeneous space.

\section{Induced representations and invariant first-order differential operator
with matrix coefficients}\label{invariant-matrix-operator}

Consider algebraic conditions for an invariant first-order linear differential operator with matrix coefficients on a homogeneous space $M$. We follow Ref. \cite{KyrnScalar} where a more general case of invariant linear matrix differential operator of the second-order was studied.

Denote by $C^{\infty}(M,V)$ and $C^{\infty}(G,V)$ the two spaces of functions that map a homogeneous space $M$ and a transformation group $G$, respectively, to a linear space $V$. The last one can be regarded as a representation space of the algebra $\mathfrak{gl}(V)$.

Functions on the homogeneous space $M$ can be considered as defined on a Lie group $G$, but invariant over the fibers $H$ of the bundle $G$ \cite{Kobyashi}. In our case, when the functions take values in a vector space $V$, the space $C^{\infty}(M,V)$ is isomorphic to a subspace of the function space
\[
\hat{\mathcal{F}}=\{\varphi\in C^{\infty}(G,V)\mid\varphi(hg)=U(h)\varphi(g),\quad h\in H\},
\]
where $U(h)$ is an exact representation of the isotropy group $H$ in $V$. For any function $\varphi\in\hat{\mathcal{F}}$, we have
\begin{equation}
\varphi(g)=\varphi(hs(x))=U(h)\varphi(s(x)),\quad g=(x,h).\label{condF}
\end{equation}
Then we can identify $\varphi(s(x))$ with a function $\varphi\in C^{\infty}(M$, $V)$. Equation (\ref{condF}) gives an explicit form of the isomorphism $\hat{\mathcal{F}}\simeq C^{\infty}(M,V)$. Differentiating relation (\ref{condF}) with respect to $h^{\alpha}$ and assuming $h=e_{H}$, we obtain
\begin{gather}\label{inf_hatF}
\left(\eta_{\alpha}+\Lambda_{\alpha}\right)\varphi(g)=0,\quad\Lambda_{\alpha}=\left.\frac{\partial U(h)}{\partial h^{\alpha}}\right|_{h=e_{H}},\quad
   \alpha=1,\dots,\mathrm{dim}H.
\end{gather}
Here, $\Lambda_{\alpha}$ are representation operators of the algebra $\mathfrak{h}$ on the space $V$. Equation (\ref{inf_hatF}) is a consequence of the condition (\ref{condF}) in the definition of $\hat{\mathcal{F}}$. The isotropy subgroup $H$ is assumed to be connected. Then the conditions (\ref{condF}) and (\ref{inf_hatF}) are equivalent.

From (\ref{inf_hatF}) we can see that a linear differential operator $R=R(g,\partial_{g})$ leaves invariant the function space $\hat{\mathcal{F}}$, if
\begin{align}
&(\eta_{\alpha}+\Lambda_{\alpha})R(g,\partial_{g})\varphi(g) 
= [\eta_{\alpha}+\Lambda_{\alpha},R(g,\partial_{g})]\varphi(g)=0,\quad\varphi\in\hat{\mathcal{F}}.\label{ravFF}
\end{align}
Thus, the space $L(\hat{\mathcal{F}})$ of linear differential operators $R(g,\partial_{g}):\hat{\mathcal{F}}\rightarrow\hat{\mathcal{F}}$ consists of linear differential operators on $C^{\infty}(G,V)$ provided that
\begin{equation}
\left.[\eta_{\alpha}+\Lambda_{\alpha},R(g,\partial_{g})]\right|_{\hat{\mathcal{F}}}=0.\label{condLF}
\end{equation}
Then given relation (\ref{condF}), the action of $R(g,\partial_{g})\in L(\hat{\mathcal{F}})$
on a function $\varphi(g)$ from the space $\hat{\mathcal{F}}$ is written
as
\begin{equation}
R(g,\partial_{g})\varphi(g)=U(h)\left(U^{-1}(h)RU(h)\right)\varphi(s(x)).\label{actR}
\end{equation}
Multiplying equation (\ref{ravFF}) by $U^{-1}(h)$ and given $\eta_{\alpha}U(h)=-\Lambda_{\alpha}U(h)$,
we obtain
\begin{gather*}
U^{-1}(h)[\eta_{\alpha}+\Lambda_{\alpha},R(g,\partial_{g})]U(h)\varphi(s(x))=
[\eta_{\alpha},U^{-1}(h)R(g,\partial_{g})U(h)]\varphi(s(x))=\\
\eta_{\alpha}\left(U^{-1}(h)R(g,\partial_{g})U(h)\varphi(s(x))\right)=0.
\end{gather*}
From here it follows that the operator $U^{-1}(h)R(g,\partial_{g})U(h)$ 
is independent of $h$ and (\ref{actR}) can be written as
\begin{align}
& R(g,\partial_{g})\varphi(g) = U(h)R_{M}(x,\partial_{x})\varphi(s(x)),\label{actRM}\\ \nonumber
& R_{M}(x,\partial_{x}) \equiv\left.\left(U^{-1}(h)R(g,\partial_{g})U(h)\right)\right|_{h=e_{H}}=\\ \nonumber
&\qquad\qquad\,\,\,=\left.R(g,\partial_{g})U(h)\right||_{h=e_{H}}.
\end{align}
That is, for any operator $R(g,\partial_{g})$ of $L(\hat{\mathcal{F}})$ there
exists an operator $R_{M}$ on the homogeneous space $M$ acting on
functions of the space $C^{\infty}(M,V)$. We say that the operator
$R_{M}(x,\partial_{x})$ is the \textit{projection} of the operator $R(g,\partial_{g})$:
$R_{M}(x,\partial_{x})=\widehat{\pi}_{*}R(g,\partial_{g})$. For example,
for a  first-order linear differential operator
\[
R_{1}(g,\partial_{g})=B^{a}(x,h)\partial_{x^{a}}+B^{\alpha}(x,h)\partial_{h^{\alpha}}+B(x,h)
\]
the projection acts as follows:
\begin{gather}
R_{M}^{(1)}(x,\partial_{x})=\widehat{\pi}_{*}R_{1}(g,\partial_{g})=
B^{a}(x,e_{H})\partial_{x^{a}}+B^{\alpha}(x,e_{H})\Lambda_{\alpha}+B(x,e_{H}).\label{piR1}
\end{gather}
On the other hand, any linear differential operator $R_{M}$ defined on $C^{\infty}(M,V)$ corresponds to an operator 
\[
R(g,\partial_{g})=U(h)R_{M}(x,\partial_{x})U^{-1}(h)\in L(\mathcal{\hat{\mathcal{F}}}).
\]
Thus, we have the isomorphism $L(\hat{\mathcal{F}})\simeq L(C^{\infty}(M,V))$ whose explicit form is given by (\ref{actRM}).

Let $\xi_{X}(g)=(L_{g})_{*}X$ be a left-invariant vector field on the Lie group $G$, where $(L_{g})_{*}:T_{g'}G\rightarrow T_{gg'}G$ is the left shift differential $L_{g}(g')=gg'$ on  $G$, $X\in \mathfrak{g}$.

Since the left-invariant vector fields commute with right-invariant ones, the condition of projectivity (\ref{condLF}) is fulfilled. Using (\ref{piR1}), we find the corresponding operator on the homogeneous space as
\begin{align}
 & \widetilde{X}(x)=\hat{\pi}_{*}\xi_{X}(g)=X(x)+\xi_{X}^{\alpha}(x,e_{H})\Lambda_{\alpha},\quad X\in\mathfrak{g},\label{xiX}\\
 & X(x)=\xi_{X}^{a}(x)\partial_{x^{a}},\quad\xi_{X}^{a}(x)=\left.\frac{d}{dt}(xe^{tX})^{a}\right|_{t=0},\label{defX}
\end{align}
where $X(x)$ are the generators of the action of the group $G$ on  $M$,  note that  $X(x)$  act in the space $C^{\infty}(M)$.  It is easy to verify the following commutation relations for operators (\ref{xiX}):
\begin{gather*}
[\widetilde{X},\widetilde{Y}]=\left.\left[U^{-1}(h)\xi_{X}U(h),U^{-1}\xi_{Y}U(h)\right]\right|_{h=e_{H}}=\\
=\left.\left[U^{-1}(h)[\xi_{X},\xi_{Y}]U(h)\right]\right|_{h=e_{H}}=
\left.\left[U^{-1}(h)\xi_{[X,Y]}U(h)\right]\right|_{h=e_{H}}=\widetilde{[X,Y]},
\end{gather*}
for all $X,Y\in\mathfrak{g}$. Consequently, the operators $\widetilde{X}$ corresponding to the left-invariant vector fields $\xi_{X}$ are \textit{generators} of a transformation group acting on $C^{\infty}(M,V)$ and are a continuation of the vector fields $ X (x) $ obeying the same commutation relations. We show how the projection of left-invariant vector fields is related to representations of the group on homogeneous space. The space $\mathcal{\hat{F}}$  invariant under right shifts on the Lie group $G$  and the operators $T_{g}$, acting by the rule $(T_{g}\varphi) (g') = \varphi (gg')$, define a representation of the group $G$ on $\mathcal{\hat{F}}$ induced by the representation $U(h)$ of the subgroup $H$. According to  relation (\ref{condF}), we have:
\begin{gather}
(T_{g}\varphi)(s(x)) =\varphi(gs(x))=
\varphi(h(x,g)s(xg))=U(h(x,g))\varphi(s(xg)),\label{TgMi-1}
\end{gather}
where $h(x,g) \in H$ is the factor of the homogeneous space \cite{Kirr}, which is determined from the system of equations
\[
s(x)g=h(x,g)s(xg),\quad h(x,e)=e_{H}.
\]
In view of the isomorphism $\hat{\mathcal{F}} \simeq C^{\infty}(M,V)$, we obtain from (\ref{TgMi-1}) a representation of the Lie group $G$ on the space of functions $\psi\in C^{\infty} (M,V)$,
\[
(T_{g}\psi)(x)=U(h(x,g))\psi(xg).
\]
This representation is called the induced representation of the group $G$ on the homogeneous space $M$. Note that 
\begin{gather*}
  \xi_{X}(g) h^{\alpha} (x,g) = \xi_{X}(g) (h(x,g) s(xg))^{\alpha} =\xi_{X} (g)(s (x) g)^{\alpha} = \xi_{X}^{\alpha} (s (x) g),
\end{gather*}
whence immediately follows the expression for the derivative of the factor at the identity element
\begin{equation}
\left.\frac{d}{dt}h^{\alpha}(x,e^{tX})\right|_{t=0}=\xi_{X}^{\alpha}(x,e_{H}).\label{he}
\end{equation}
It is easy to see that the operators $\widetilde {X}(x)$,  as described by (\ref{defX}) and (\ref{he}),
  are  differentials of the representation $T_{g}$ on the homogeneous space $M$:
  \[
\widetilde{X}(x)\psi(x)=\left.\frac{d}{dt}\left(T_{\exp(tX)}\psi\right)(x)\right|_{t=0}.
\]

Thus, the projection of left-invariant vector fields on the group gives  the infinitesimal operators of the representation of $T_{g}$ induced by the representation $U(h)$ of the subgroup $H$.

An operator $R_{M}(x,\partial_{x})\in L(C^{\infty}(M,V))$ is invariant under the action of the Lie group of transformations, if $R_{M}(x,\partial_{x})$ commutes with $\widetilde{X}$:
\begin{gather*}
	[R_{M}(x,\partial_{x}),\widetilde{X}]=
	[U^{-1}(h)RU(h),U^{-1}(h)\xi_{X}U(h)]|_{h=e_{H}}=\\
   =U^{-1}(h)[R(g,\partial_{g}),\xi_{X}]U(h)|_{h=e_{H}}=0.
\end{gather*}
It follows that the operator $R_{M}(x,\partial_{x})$ is invariant
with respect to the transformation group if and only if the corresponding
operator $R(g,\partial_{g})\in L(\hat{\mathcal{F}})$ commutes with
the left-invariant vector fields:
\begin{equation}
[R(g,\partial_{g}),\xi_{X}]=0,\quad X\in\mathfrak{g}.\label{commXi}
\end{equation}
Let $R_{M}^{(1)}(x,\partial_{x})\in C^{\infty}(M,V)$ be a linear
differential operator of the first order, invariant with respect to
the group action. By (\ref{commXi}), this operator corresponds to
a first-order polynomial of right-invariant vector fields:
\[
R_{(1)}(g,\partial_{g})=B^{a}\eta_{a}(x,h)+B^{\alpha}\eta_{\alpha}(h)+\tilde{B}.
\]
As a result of the projection, the expression $B^{\alpha}\eta_{\alpha}(h)$
becomes constant $B^{\alpha}\Lambda_{\alpha}$, which can be eliminated
in the operator $R_{M}^{(1)}(x,\partial_{x})$ by changing the variable
$B=\tilde{B}+B^{\alpha}\Lambda_{\alpha}$. Therefore, we can put $B^{\alpha}=0$
without loss of generality. If we substitute the operator $R_{(1)}(g,\partial_{g})$
in the projectivity condition (\ref{condLF}), then we obtain
\begin{gather*}
\left.[\eta_{\alpha}+\Lambda_{\alpha},R_{(1)}(g,\partial_{g})]\right|_{\hat{\mathcal{F}}}=\left([b^{a},\Lambda_{\alpha}]\eta_{a}+b^{a}[\eta_{a},\eta_{\alpha}]+[B,\Lambda_{\alpha}]\right)_{\hat{\mathcal{F}}}=\\
\left([B^{a},\Lambda_{\alpha}]+B^{b}C_{b\alpha}^{a}\right)\left.\eta_{a}\right|_{\hat{\mathcal{F}}}+[B,\Lambda_{\alpha}]-B^{a}C_{a\alpha}^{\beta}\Lambda_{\beta}=0.
\end{gather*}
Also we have a system of algebraic equations for the coefficients
$B^{a}$ and $B$:
\begin{align}
 & [B^{a},\Lambda_{\alpha}]+B^{b}C_{b\alpha}^{a}=0,\label{sysBa}\\
 & [B,\Lambda_{\alpha}]-B^{a}C_{a\alpha}^{\beta}\Lambda_{\beta}=0.\label{sysB}
\end{align}
When equations (\ref{sysBa}) – (\ref{sysB}) are fulfilled, the projection of $R_{(1)}(g,\partial_{g})$ on the homogeneous space results in the desired form of the invariant linear differential operator of the first order:
\begin{gather}
R_{M}^{(1)}(x,\partial_{x})=\hat{\pi}_{*}R_{(1)}(g,\partial_{g})= B^{a}\eta_{a}^{i}(x,e_{H})\partial_{x^{i}}+B^{a}\eta_{a}^{\alpha}(x,e_{H})\Lambda_{\alpha}+B.\label{R1form}
\end{gather}
So, any linear differential operator of the first order acting on the functions of $C^{\infty}(M,V)$ and being invariant with respect to the action of the transformation group has the form (\ref{R1form}) where the matrix coefficients $B^{a}$ and $B$ satisfy the algebraic system of equations (\ref{sysBa}) –- (\ref{sysB}). The matrices $\Lambda_{\alpha}$ are generators of the isotropy subgroup $H$ in a linear space $V$.

\section{$\lambda$-representation of a Lie algebra}\label{sec:lambda}

 In this section we describe a special representation of the Lie algebra $\mathfrak {g}$ using the orbit method \cite{Kirr}. The direct and inverse Fourier transforms on the Lie group $G$ are introduced, that in what follows are necessary for the noncommutative reduction of the Dirac equation on the homogeneous space $M$. Here we also use some results of the previous section.

First, we describe an orbit  classification for the coadjoint representation of Lie groups following  conventions of Refs. \cite{ShDarbu,ShT}.

A degenerate Poisson--Lie bracket,
\begin{gather}
\left\{ \phi,\psi\right\} (f) = \langle f,\left[d\phi(f),d\psi(f)\right]\rangle=C_{AB}^{C}f_{C}\frac{\partial\phi(f)}{\partial f_{A}}\frac{\partial\psi(f)}{\partial f_{B}},\quad\phi,\psi\in C^{\infty}(\mathfrak{g}^{*}),\label{pl1}
\end{gather}
endows the space $\mathfrak{g}^{*}$ with a Poisson structure.
Here $f_{A}$ are coordinates of a linear functional $f=f_{A}e^{A} \in\mathfrak{g}^{*}$ relative to the dual basis $\left\{e^{A}\right\}$. The number $\mathrm{ind\,\mathfrak{g}}$ of functionally independent Casimir functions $K_{\mu}(f)$ relative to the bracket (\ref{pl1}) is called the \textit {index} of the Lie algebra $\mathfrak{g}$.

A \textit{coadjoint} representation on $\mathfrak{g}^{*}$,
$\mathrm{Ad}^{*}$: $G\times\mathfrak{g}^{*}\rightarrow \mathfrak{g}^{*}$, stratifies $\mathfrak{g}^{*}$ into orbits of the coadjoint representation (K-orbits). The restriction of the bracket (\ref{pl1}) on  orbits  is non-degenerate and coincides with the Poisson bracket generated by the symplectic Kirillov form $\omega_{\lambda}$.

The orbits of maximal dimension $\mathrm{dim}\, \mathcal{O}^{(0)}=\mathrm{dim}\, \mathfrak{g} - \mathrm{ind}\, \mathfrak{g}$ are called \textit{non-degenerate}, and the those of less dimension are singular. We denote by $\mathcal{O}_{\lambda}^{(s)}$ the orbits of dimension $\mathrm{dim}\,\mathfrak{g} - \mathrm{ind}\,\mathfrak{g}-2s$, $s = 0, \dots, (\mathrm{dim}\, \mathfrak{g} - \mathrm{ind}\, \mathfrak{g})/ 2$ passing through the functional $\lambda\in\mathfrak{g}^{*}$, and a number $s$ is called the \textit{orbit singularity index}. A tangent space $T_{f}\mathcal{O}_{\lambda^{(s)}}$ to the orbit $\mathcal{O_{\lambda}}^{(s)}$ at a point $f$ is the linear span of vector fields
\[
Y_{A}(f)=C_{AB}(f)\frac{\partial}{\partial f_{B}},\quad C_{AB}(f)=C_{AB}^{C}f_{C},
\]
so that the orbit dimension is given by the rank of the matrix $C_{AB}(f)$. The rank takes a constant value on the orbit, $\mathrm{dim\,\mathcal{O}}_{\lambda}^{(s)}=\mathrm{rank}\,C_{AB}(\lambda)$. The Kirillov form on tangent vectors to the orbit $\mathcal{O}_{\lambda}^{(s)}$ is  $\omega_{\lambda} (Y_{A}, Y_{B}) = \langle\lambda, [Y_{A}, Y_{B}] \rangle$.

The space $\mathfrak{g}^{*}$ can be decomposed into a sum of disjoint invariant algebraic surfaces $M_{s}$ consisting of orbits of the same dimension $\mathrm{dim}\,\mathfrak{g}-\mathrm{ind}\,\mathfrak{g}-2s$:
\begin{align*}
 & M_{0}=\left\{ f\in\mathfrak{g}^{*}\mid\lnot\left(F^{1}(f)=0\right)\right\} ,\\
 & M_{s}=\left\{ f\in\mathfrak{g}^{*}\mid F^{s}(f)=0,\lnot\left(F^{s+1}(f)=0\right)\right\} ,\quad s=1,\dots,\frac{\mathrm{dim}\,\mathfrak{g}-\mathrm{ind}\,\mathfrak{g}}{2}-1,\\
 & M_{(\mathrm{dim}\,\mathfrak{g}-\mathrm{ind}\,\mathfrak{g})/2}=\left\{ f\in\mathfrak{g}^{*}\mid F^{\frac{\mathrm{dim}\mathfrak{g}-\mathrm{ind}\mathfrak{g}}{2}}(f)=0\right\} ,
\end{align*}
where $F^{s}(f)$ denotes the set of all minors of the matrix $C_{AB}(f)=C_{AB}^{C}f_{C}$ of size $\mathrm{dim}\,\mathfrak {g } -\mathrm{ind}\,\mathfrak{g} -2s+2$; the notation $F^{s}(f) = 0$ implies that all the corresponding minors at the point $f$ vanish, and $\lnot(F^{s}(f) = 0)$ means that at the point $f$, the corresponding minors do not vanish simultaneously. In the general case, the surface $M_{s}$ is disconnected.

In what follows, by $M_{(s)}$ we  denote a connected component of the surface $M_{s}$ containing the orbit $\mathcal{O}^{(s)}$, and  $\mathcal{O}^{(s)}$ will be called the $(s)$-type orbit. Each component of $M_{(s)}$ is uniquely determined by a set of homogeneous polynomials $F_{\alpha}^{(s)}(f)$ satisfying the system
\[
\left.Y_{A}(f)F_{\alpha}^{(s)}(f)\right|_{F^{(s)}(f)=0}=0.
\]

The non-constant functions $K_{\mu}^{(s)}(f)$ on $M_{(s)}$  are called the  $(s)$-type  Casimir functions if they commute with any function on $M_{(s )}$ with respect to the bracket  (\ref{pl1}). In other words, the functions $ K_{\mu}^{(s)}(f)$ are invariants of the adjoint representation and are determined by the system
\[
\left.Y_{A}(f)K_{\mu}^{(s)}(f)\right|_{f\in M_{(s)}}=0,\quad\mu=1,\dots,r_{(s)}.
\]
The number of functionally independent solutions of this system is determined by the dimension of the surface $M_{(s)}$:
\[
r_{(s)}=\mathrm{dim}\,M_{(s)}-(\mathrm{dim}\,\mathfrak{g}-\mathrm{ind}\,\mathfrak{g}-2s).
\]
Denote by $\Omega^{(s)}\subset\mathbb{R}^{r_{(s)}}$ a set of values of the mapping $K^{(s)}: M_{(s)}\rightarrow\mathbb {R}^{r_{(s)}}$ and introduce a locally invariant subset
\begin{gather*}
\mathcal{O}_{\omega}^{(s)}=\bigg{\{} f\in M_{(s)}\mid K_{\mu}^{(s)}(f)=\omega_{\mu}^{(s)},\quad\mu=1,\dots,r_{(s)},\quad\omega^{(s)}\in\Omega^{(s)}\bigg{\}}.
\end{gather*}
If the Casimir functions $K_{\mu}^{(s)}(f)$ are single-valued, then the level surface $\mathcal{O}_{\omega}^{(s)}$ consists of a countable set of orbits. We call $\mathcal{O}_{\omega}^{(s)}$ the \textit{class of orbits}.
As a result, the space $\mathfrak{g}^{*}$ consists of a union of connected invariant disjoint algebraic surfaces $M_{(s)}$, which in turn is the union of the orbit classes $\mathcal{O}_{\omega}^{(s)}$:
\begin{equation}
\mathfrak{g}^{*}=\bigcup_{(s)}M_{(s)}=\bigcup_{(s)}\bigcup_{\omega^{(s)}\in\Omega^{(s)}}\mathcal{O}_{\omega}^{(s)}.\label{rg}
\end{equation}

Consider a quotient space $B_{(s)}=M_{(s)}/G$, $\mathrm{dim}B_{(s)}=r_{(s)}$, whose points are the orbits of one class, $\mathcal{O}_{\lambda}^{(s)}\in M_{(s)}$.
We introduce a local section $\lambda (j)$ of the bundle $M_{(s)}$ with base $B_{(s)}$ using  real parameters $ j = (j_{1}, \dots, j_{r_{( s)}}) $ taking their values in a domain $J\subset\mathbb {R}^{r_{(s)}}$:
\begin{gather*}
F^{(s)}(\lambda(j))=0,\quad K_{\mu}^{(s)}(\lambda(j))=\omega_{\mu}^{(s)}(j),\quad
\mathrm{det}\left\Vert \frac{\partial\omega_{\mu}^{(s)}(j)}{\partial j_{\nu}}\right\Vert \neq0.
\end{gather*}
Let $\mathcal{O}_{\lambda (j)}^{(s)}$ be a K-orbit of $(s)$-type passing through a covector $\lambda=\lambda (j)\in\mathfrak {g}^{*}$ and belonging to the same class of orbits for all $j\in J$.

Using the Kirillov orbit method \cite{Kirr}, we construct  a unitary irreducible representation of the Lie group $G$ on a given orbit. This representation can be constructed if and only if for the functional $\lambda$ there exists a subalgebra $\mathfrak{p\subset\mathfrak{\mathfrak{g}^{\mathbb{C}}}}$ in the complex extension $\mathfrak{g}^{\mathbb{C}}$ of the Lie algebra $\mathfrak{g}$ satisfying the conditions:
\begin{equation}
\langle\lambda,[\mathfrak{p},\mathfrak{p}]\rangle=0,\quad\mathrm{dim}\,\mathfrak{p}=\mathrm{dim}\,\mathfrak{g}-\frac{1}{2}\mathrm{dim}\,\mathcal{O}_{\lambda}^{(s)}.\label{defp}
\end{equation}

The subalgebra $\mathfrak{p}$ is called the polarization of the functional $\lambda$. In (\ref{defp}), it is assumed that the functionals from the space $\mathfrak{g}^{*}$ are extended to $\mathfrak{g}^{\mathbb{C}}$ by linearity. Moreover, real polarizations always exist for nilpotent and completely solvable Lie algebras, and the complex polarizations always exist for solvable Lie groups \cite{Dix}. For non-degenerate orbits $\mathcal{O}_{\lambda}^{(0)}$ there always exists, generally speaking, a complex polarization. In this paper, for  simplicity, we restrict ourselves to the case when $\mathfrak{p}$ is the real polarization.

Denote by $P$ a closed subgroup of the Lie group $G$ whose Lie algebra is $\mathfrak{p}$. The Lie group acts on the right homogeneous space $Q\simeq G/P $: $ q '= qg$. According to the orbit method, we introduce a unitary one-dimensional irreducible representation of the Lie group $P$, which, in the neighborhood of the identity element of the group,  has the form
\begin{equation}
U^{\lambda}(e^{X})=\exp\left(\frac{i}{\hbar}\langle\lambda,X\rangle\right),\quad X\in\mathfrak{p}.\label{Ul}
\end{equation}

The representation of the Lie group $G$ corresponding  to the orbit $\mathcal{O}_{\lambda}^{(s)}$ is induced using \eqref{Ul} as
\begin{align}
 & (T_{g}^{\lambda}\psi)(q)=\sqrt{\frac{\Delta_{G}(p(q,g))}{\Delta_{P}(p(q,g))}}U^{\lambda}(p(q,g))\psi(qg)= U^{\lambda+i\hbar\beta}(p(q,g))\psi(qg),\label{TgG}\\
 & \beta=\left.d\log\sqrt{\frac{\Delta_{G}(p)}{\Delta_{H}(p)}}\right|_{u=e_{P}},\nonumber
\end{align}
where $\Delta_{G}(g)=\mathrm{det}\mathrm{^{- 1} Ad}_{g}$ is the module of the Lie group $G$,   $\Delta_{P}(p) = \mathrm{det}\mathrm{Ad}_{p}$ is the module of the subgroup $P$, $p\in P$, and  $e_{P}$ is the identity element of
$P$. A function $p(q,g)$ is the factor of the homogeneous space $Q$.

The functions $\psi^{\lambda}(q;g)=(T_{g}^{\lambda}\psi)(q)$ on the  group $G$ satisfy a condition similar to (\ref{condF}):
\[
\psi^{\lambda}(q;pg)=U^{\lambda+i\hbar\beta}(p)\psi^{\lambda}(q;g),\quad u\in P.
\]
The space of all such functions will be denoted by $\mathcal{F^{\lambda}}$. 
Restriction of the left-invariant vector fields $\xi_{X}(g)$ to a homogeneous space $Q$, as follows from results of
 section \ref{invariant-matrix-operator}, is correctly defined, and the explicit form of the corresponding operator on the homogeneous space is given by  (\ref{actRM}):
 \begin{align}
 & \ell_{X}(q,\partial_{q},\lambda)=\left.\left(\left[U^{\lambda+i\hbar\beta}(p)\right]^{-1}\xi_{X}(g)U^{\lambda+i\hbar\beta}(p)\right)\right|_{p=e_{P}},\label{getEll-1}\\
 & [\ell_{X}(q,\partial_{q},\lambda),\ell_{Y}(q,\partial_{q},\lambda)]=\ell_{[X,Y]}(q,\partial_{q},\lambda),\quad X,Y\in\mathfrak{g}.\label{defL0}
\end{align}

Equation (\ref{getEll-1}) shows that $\ell_{X}(q,\partial_{q},\lambda)$ are  infinitesimal operators of the induced  representation (\ref{TgG}),
\[
\ell_{X}(q,\partial_{q},\lambda)\psi(q)=\left.\frac{d}{dt}(T_{\exp(tX)}^{\lambda}\psi)(q)\right|_{t=0}.
\]

Denote by $L(Q, \mathfrak {h}, \lambda)$ a space of functions on $Q$  where  representation (\ref{TgG}) is defined. The representation (\ref{TgG}) is unitary with respect to a scalar product of the function space $L_{2}(Q, \mathfrak {h}, \lambda)$:
\begin{gather}
(\psi_{1},\psi_{2}) = \int_{Q}\overline{\psi_{1}(q)}\psi_{2}(q)d\mu(q),\quad 
d\mu(q)=\rho(q)dq^{1}\dots dq^{\dim Q}.\label{scQ}
\end{gather}
The function $\rho(q)$ is determined from the Hermitian condition for the operators $-i\ell_{X}(q,\partial_{q},\lambda)$ with respect to this scalar product (\ref{scQ}).

The irreducible representation of the Lie algebra $\mathfrak {g}$ by the linear operators of the first order (\ref{getEll-1}) dependent on $\mathrm{dim}\, \mathcal{O}_{\lambda}^{(s)}/2$ variables is called $\lambda$ - \textit{representation} of the Lie algebra $\mathfrak{g}$ and it was  introduced in Ref. \cite{SpSh1}.

Let us  describe a recipe for calculating  operators of the $\lambda$-representation on the class of orbits $\mathcal{O}_{\omega}^{(s)}$. We find a local section $\lambda (j)$ of the bundle $M_{(s)}$ with the base $B_{(s)}$ for which there exists a polarization $\mathfrak {p}$ with a basis $\left\{e_{\overline {\alpha}} \right\} $.  Let us then fix the basis $\left\{e'_{\overline {a}}\right\}$ in the additional subspace $\mathfrak{m = \mathfrak{p}^{\perp}}$ to the subalgebra $\mathfrak{ p}$. In a neighborhood of the identity element  of the Lie group $G$, we introduce  local coordinates of the second kind,
\begin{gather*}
g(p,q)=\left(e^{p^{\dim\mathfrak{p}}e_{\overline{\dim\mathfrak{p}}}}e^{p^{\dim\mathfrak{p}-1}e_{\overline{\dim\mathfrak{p}-1}}}\dots e^{p^{1}e_{\overline{1}}}\right)
\left(e^{q^{\dim Q}e'_{\overline{\dim Q}}}e^{q^{\dim Q-1}e'_{\overline{\dim Q-1}}}\dots e^{q^{1}e'_{\overline{1}}}\right),	
\end{gather*}
and write the left-invariant vector fields in local coordinates $(p,q)$:
\[
\xi_{X}(p,q)=\xi_{X}^{\overline{a}}(q)\partial_{q^{\overline{a}}}+\xi_{X}^{\overline{\alpha}}(q,p)\partial_{p^{\overline{\alpha}}}.
\]
 Next, operators of the $\lambda$ representation are defined using
(\ref{getEll-1}) as
\begin{align*}
 & \ell_{X}(q,\partial_{q},\lambda)=\xi_{X}^{\overline{a}}(q)\partial_{q^{\overline{a}}}+\frac{i}{\hbar}\xi_{X}^{\overline{\alpha}}(q,e_{P})\left(\lambda_{\overline{\alpha}}+i\hbar\beta_{\overline{\alpha}}\right),\\
 & \beta_{\overline{\alpha}}=\frac{1}{2}\left[\mathrm{Tr}\left(\mathrm{ad}_{\overline{\alpha}}\right)-\mathrm{Tr}\left(\left.\mathrm{ad}_{\overline{\alpha}}\right|_{\mathfrak{p}}\right)\right].
\end{align*}
In other words, finding of these operators is reduced to calculating the left-invariant vector fields on the group $G_{\mathbb {}}$ in  the trivialization domain of the  principal bundle of this group in the fibrations $P=\exp(\mathfrak {p})$.

Let us write the representation operators (\ref{TgG}) in integral form
\begin{eqnarray}
  (T_{g}^{\lambda}\psi)(q) & = &\int_{Q}\psi(q'){D}_{qq'}^{\lambda}(g)d\mu(q),\label{funkD}\\ \nonumber
  {D}_{qq'}^{\lambda}(g)&= &\sqrt{\frac{\Delta_{G}(p(q,g))}{\Delta_{H}(p(q,g))}}U^{\lambda}(p(q,g))\delta(qg,q'),
\end{eqnarray}
where $\delta(q,q')$ is the generalized delta-function with respect to the measure $d\mu(q)$. The generalized kernels $ {D}_{qq'}^{\lambda}(g)$ of this representation satisfy the following properties:
\begin{gather}\nonumber
 {D}_{qq'}^{\lambda}(g_{1}g_{2})=\int_{Q} {D}_{qq''}^{\lambda}(g_{1}) {D}_{q''q'}^{\lambda}(g_{2})d\mu(q''),\\ \nonumber
 {D}_{qq'}^{\lambda}(g)=\overline{ {D}_{q'q}^{\lambda}(g^{-1})},\quad {D}_{qq'}^{\lambda}(e)=\delta(q,q'),
\end{gather}
and  the system of equations
\begin{gather}
\left(\eta_{X}(g)+\ell_{X}(q,\partial_{q},\lambda)\right) {D}_{qq'}^{\lambda}(g)=0,\nonumber \\ 
\left(\xi_{X}(g)+\overline{\ell_{X}(q',\partial_{q'},\lambda)}\right) {D}_{qq'}^{\lambda}(g)=0.\label{tD}
\end{gather}
Here $g_{1}, g_{2}\in G$.
Let $G^{\lambda} = \left\{g \in G \mid Ad_{g}^{*} \lambda = \lambda\right \}$ be a stabilizer of the functional $\lambda\in\mathfrak{g }^{*}$ with the Lie algebra $\mathfrak {g}^{\lambda}$, and the vectors
$\{e_{\alpha},  \alpha = 1, \dots, \dim\mathfrak {g}^{\lambda} \}$ form  some fixed basis for  $\mathfrak {g}^{\lambda} \subset \mathfrak {p}$. As $P$ is  the stabilizer of the point $q=0$ in the homogeneous space $Q$, and the group $G^{\lambda}$ lies in $P$, we get  the equality
\[
\ell_{\alpha}(0,\partial_{q},\lambda)=\frac{i}{\hbar}\left(\lambda_{\alpha}+i\hbar\beta_{\alpha}\right),\quad\alpha=1,\dots,\dim\mathfrak{g}^{\lambda}.
\]
Restricting the first equality (\ref{tD}) to the subgroup $G^{\lambda}$ and setting $q=0$, we find
\begin{equation}
\left(\eta_{\alpha}(h)+\frac{i}{\hbar}\left(\lambda_{\alpha}+i\hbar\beta_{\alpha}\right)\right) {D}_{0q'}^{\lambda}(h)=0,\quad h\in G^{\lambda}.\label{gl1}
\end{equation}

The solution of the system \eqref{gl1} up to a constant factor can be represented as
\begin{equation}
 {D}_{0q'}^{\lambda}(h)=\exp\left(-\frac{i}{\hbar}\int\sigma^{\lambda}(h)+\int\sigma^{\overline{\alpha}}(h)\beta_{\overline{\alpha}}\right).\label{solGl2}
\end{equation}

The subalgebra $\mathfrak {g}^{\lambda}$ is subordinate to the covector $\lambda$, and the 1-forms $\sigma^{\lambda} (h)$ and $\sigma^{\overline{\alpha}}(h) \beta_{\overline{\alpha}}$ are closed in  $G^{\lambda}$. Thus, the integral in (\ref{solGl2}) is well-defined. The local solution \eqref{solGl2} can be extended to a  global one,  if the integral on the right-hand side of \eqref{solGl2} over any closed curve $\Gamma$ on the subgroup $G^{\lambda}$ is a multiple of $2\pi i$. Note that since the 1-form $\sigma^{\lambda}(h)$ is closed, the value of this integral depends only on the homological class to which the curve  $\Gamma$ belongs. Therefore, for a global solution of the system \eqref{gl1}, the  following condition should be satisfied:
\begin{equation}
\frac{1}{2\pi\hbar}\oint\limits _{\Gamma\in H_{1}(G^{\lambda})}\sigma^{\lambda}(h)=n_{\Gamma}\in\mathbb{Z}.\label{conjO}
\end{equation}
In other words, the 1-form $\sigma^{\lambda}(h)$ should belong to an integral cohomology class from $H^{1}(G^{\lambda}, \mathbb{Z})$. In the case of a simply connected group $G$, the condition \eqref{conjO} is equivalent to the \textit{ condition of integral orbit} $\mathcal {O}_{\lambda}^{(s)}$, proposed by A.~A.~Kirillov \cite{Kirr}:
\begin{equation}
\frac{1}{2\pi\hbar}\int\limits _{\Gamma\in H^{2}(\mathcal{O}_{\ensuremath{\lambda}})}\ensuremath{\omega_{\lambda}}=\ensuremath{n_{\Gamma}\in}\ensuremath{\mathbb{Z}}.\nonumber
\end{equation}
Thus, for a simply connected group the coadjoint orbit   $\mathcal{O}_{\lambda}^{(s)}$ is integral if  the equality \eqref{conjO} is fulfilled.

A set of generalized functions $ {D}_{qq'}^{\lambda}(g)$ satisfying the system  (\ref{tD}) was studied in Refs. \cite{SpSh1, ShDarbu} and the hypothesis was proposed that
 this set of generalized functions has the properties of completeness and orthogonality for a certain choice of the measure $d\mu(\lambda)$ in the parameter space $J$:
\begin{gather}
\int_{G}\overline{ {D}_{\widetilde{q}\widetilde{q}'}^{\widetilde{\lambda}}(g)} {D}_{qq'}^{\lambda}(g)d\mu(g)=\Delta_{G}(s(q))\delta(q,\tilde{q})\delta(\tilde{q}',q')\delta(\tilde{\lambda},\lambda),\label{Dort}\\
\int_{Q\times Q\times J}\overline{ {D}_{qq'}^{\lambda}(\tilde{g})} {D}_{qq'}^{\lambda}(g)d\mu(q)\frac{d\mu(q')}{\Delta_{G}(s(q))}d\mu(\lambda)=\delta(\tilde{g},g).\label{Dful}
\end{gather}
Here $\delta(g)$ is the generalized Dirac delta function with respect to the left Haar measure $d\mu(g)$ on the Lie group $G$.

For compact Lie groups, the relations (\ref{Dort})-(\ref{Dful}) hold by the Peters-Weil theorem \cite{barut}. Note that although there is no rigorous proof of the relations (\ref{Dort})-(\ref{Dful}), in each case it is easy to verify directly their validity.

Consider the  space $L(G,\lambda, d\mu(g))$ of functions of the form
\begin{equation}
\psi^{\lambda}(g)=\int_{Q}\psi(q,q',\lambda) {D}_{qq'}^{\lambda}\left(g^{-1}\right)d\mu(q')d\mu(q).\label{psiLD}
\end{equation}
Here,  a function $\psi(q,q', \lambda)$ of the two variables  $q$ and $q'$ belongs to the space $L(Q, \mathfrak{h}, \lambda)$.
The inverse transform reads
\begin{gather}
\psi(q,q',\lambda)=\Delta_{G}^{-1}(s(q'))\int_{Q\times Q\times J}\psi_{\lambda}(g)\overline{ {D}_{qq'}^{\lambda}\left(g^{-1}\right)}d\mu_{R}(g), \label{invF}
\end{gather}
where  we have used  (\ref{Dort})--(\ref{Dful}),  and  $d\mu_{R}(g)=d\mu(g^{- 1})$ is the right Haar measure on the Lie group $ G $.

The action of the operators $\xi_{X}(g)$ and $\eta_{X}(g)$ on the function $\psi^{\lambda} (g)$ from $L_{2}(G, \lambda, d\mu (g))$, according to (\ref{psiLD}) and (\ref{invF}), corresponds to  action of the operators $\overline{\ell_{X}^{\dagger} (q, \partial_{q }, \lambda)}$ and $\ell_{X} (q', \partial_{q'}, \lambda)$ on the function $\psi (q, q', \lambda)$ respectively,
\begin{align*}
 & \xi_{X}(g)\psi^{\lambda}(g)\Longleftrightarrow\overline{\ell_{X}^{\dagger}(q,\partial_{q},\lambda)}\psi(q,q',\lambda)\text{,}\\
 & \eta_{X}(g)\psi^{\lambda}(g)\Longleftrightarrow\ell_{X}(q',\partial_{q'},\lambda)\psi(q,q',\lambda).
\end{align*}

The functions (\ref{psiLD}) are eigenfunctions  for the Casimir operators $K_{\mu}^{(s)} (i \hbar\xi) = K_{\mu}^{(s)} (- i \hbar\eta)$. Indeed, from the system (\ref{tD}) we can obtain
\begin{align*}
 & K_{\mu}^{(s)}(i\hbar\xi)\psi^{\lambda}(g)=K_{\mu}^{(s)}(-i\hbar\eta)\psi^{\lambda}(g)=\\
 & \int_{Q}\left(\overline{K_{\mu}^{(s)}(-i\hbar\ell(q,\partial_{q},\lambda))}\psi(q,q',\lambda)\right) {D}_{qq'}^{\lambda}\left(g^{-1}\right)d\mu(q')=\\
 & \int_{Q}\left(K_{\mu}^{(s)}(-i\hbar\ell(q',\partial_{q'},\lambda))\psi(q,q',\lambda)\right) {D}_{qq'}^{\lambda}\left(g^{-1}\right)d\mu(q').
\end{align*}
It follows that the operators $K_{\mu}^{(s)} (-i\hbar\ell(q',\lambda))$ are independent of  $q'$ and
\begin{align*}
 & K_{\mu}^{(s)}(i\hbar\xi)\psi^{\lambda}(g)\Longleftrightarrow\kappa_{\mu}^{(s)}(\lambda)\psi(q,q',\lambda),\\
 & K_{\mu}^{(s)}(-i\hbar\ell(q',\partial_{q'},\lambda))=\kappa_{\mu}^{(s)}(\lambda),\quad\overline{\kappa_{\mu}^{(s)}(\lambda)}=\kappa_{\mu}^{(s)}(\lambda),\\
 &\lim_{\hbar\rightarrow0}\kappa_{\mu}^{(s)}(\lambda)=\omega_{\mu}^{(s)}(\lambda).
\end{align*}
Thus, as a result of the generalized Fourier transform (\ref{psiLD}), the left and the right fields become  the operators of $\lambda$-representations, and the Casimir operators become constants. This fact is a key point for the method of noncommutative integration of linear partial differential equations on Lie groups, since it allows one to reduce  the original differential equation with $\dim G$  independent variables to an equation with fewer independent variables equal to $\dim Q$.

\section{Dirac equation in homogeneous space}\label{dirac-eq-in homogen-space}

In this section, we consider the Dirac equation in a $n$-dimensional
homogeneous space $M$ with an invariant metric. We shall assume that
in the homogeneous space $M$ an invariant metric $g_{M}$ and the
Levi-Civita connection are given. Denote by $V_{\Psi}$ a space
of spinor fields on $M$.

We write the Dirac equation in the space $M$ as an equation in a
$n$-dimensional Lorentz manifold $M$ with the metric
($\hbar$ is the Planck constant) as follows \cite{BagrQG}:
\begin{gather}
\left(i\hbar\gamma^{i}(x)[\nabla_{i}+\Gamma_{i}(x)]-m\right)\psi(x)=0,\label{diracM}\quad
i=1,\dots,n.
\end{gather}
Here $\nabla_{i}$ is the covariant derivative corresponding to the
Levi-Civita connection on $M$, $m$ is  mass of the field $\psi\in C^{\infty}(M,V_{\Psi})$,
$\psi(x)$ is a column with $2^{\left\lfloor n/2\right\rfloor }$
components, $\gamma^{i}(x)$ are $2^{\left\lfloor n/2\right\rfloor }\times2^{\left\lfloor n/2\right\rfloor }$
gamma matrices,
\begin{equation}
\{\gamma_{i}(x),\gamma_{j}(x)\}=2g_{ij}(x)E,
\quad i,j=1,\dots,n,\label{sys_gamma}
\end{equation}
where $E$ denotes the $2^{\left\lfloor n/2\right\rfloor }\times2^{\left\lfloor n/2\right\rfloor }$ identity matrix,
$\Gamma_{i}(x)$ is the spinor connection satisfying  the conditions
$[\nabla_{i},\gamma_{j}(x)]=0,\ \operatorname*{Tr}\Gamma_{i}(x)=0$.
The spinor connection $\Gamma_{i}(x)$  can be written as follows \cite{BagrQG}:
\begin{equation}
\Gamma_{i}(x)=-\frac{1}{4}(\nabla_{i}\gamma_{k}(x))\gamma^{k}(x).
\label{sys_gamma-a}
\end{equation}
We seek a solution of (\ref{sys_gamma}) with the decomposition
\begin{equation}
\gamma^{i}(x)=\hat{\gamma}^{a}\eta_{a}^{i}(x,e_{H}),\quad\hat{\gamma}^{a}=\gamma^{i}(x)\sigma_{i}^{a}(x,e_{H}).\label{my_sol_gamma}
\end{equation}
The constant matrices $\hat{\gamma}^{a}$ satisfy the algebraic equations
\begin{equation}
\{\hat{\gamma}^{a},\hat{\gamma}^{b}\}=2G^{ab}E,
\quad a,b=1,\dots,n.\label{sys_gamma2}
\end{equation}
For the Dirac matrices with subscripts using (\ref{gij_loc}) we have
\begin{equation}
\gamma_{i}(x)=g_{ij}(x)\gamma^{i}(x)=\hat{\gamma}_{a}\sigma_{i}^{a}(x,e_{H}),\quad\hat{\gamma}_{b}\equiv G_{ab}\hat{\gamma}^{a}.\label{gamma_down}
\end{equation}
The spinor connection is given by the following Lemma.
\begin{lemma}
\label{prop1}The spinor connection $\Gamma(x)=\gamma^{i}(x)\Gamma_{i}(x)$
on the homogeneous space $M$ with invariant metric $g_{M}$ reads
\begin{gather}
\Gamma(x)=\hat{\gamma}^{a}\left(\Gamma_{a}+\eta_{a}^{\alpha}(x,e_{H})\Lambda_{\alpha}\right),\quad
\Gamma_{a}=-\frac{1}{4}\Gamma_{ba}^{d}\hat{\gamma}^{b}\hat{\gamma}_{d},\nonumber\\ 
\Lambda_{\alpha}=-\frac{1}{8}G_{ac}C_{\alpha b}^{a}[\hat{\gamma}^{b},\hat{\gamma}^{c}].\label{sp_gamma_x}
\end{gather}
\end{lemma}
\begin{proof}
The function $\Gamma(x)=\gamma^{i}(x)\Gamma_{i}(x)$ with   $\Gamma_{i}(x)$ given by \eqref{sys_gamma-a},
can be written as
\begin{equation}
\Gamma(x)=\frac{1}{4}\gamma^{i}(x)\gamma^{k}(x)\left(\partial_{x^{i}}\gamma_{k}(x)-\Gamma_{ki}^{l}(x)\gamma_{l}(x)\right), \label{connection-1}
\end{equation}
where $\Gamma_{ki}^{l}(x)$ are the Christoffel symbols, and $\partial_{x^{i}}$ is the partial derivative.
Substituting (\ref{gamma_P_IJK}), (\ref{my_sol_gamma}), and (\ref{gamma_down}) in (\ref{connection-1}), we obtain
\[
\Gamma(x)=\Gamma+\frac{1}{4}C_{b\alpha}^{d}\hat{\gamma}^{a}\hat{\gamma}^{b}\hat{\gamma}_{d}\sigma_{i}^{\alpha}(x,e_{H})\eta_{a}^{i}(x,e_{H}).
\]
Using property (\ref{conds_AdH}) of the invariant metric, we reduce the
expression $C_{b\alpha}^{d}\hat{\gamma}^{b}\hat{\gamma}_{d}$ to the
form
\begin{gather*}
C_{b\alpha}^{d}\hat{\gamma}^{b}\hat{\gamma}_{d}=C_{b\alpha}^{d}G_{dc}\hat{\gamma}^{b}\hat{\gamma}^{c}=-C_{c\alpha}^{d}G_{db}\hat{\gamma}^{b}\hat{\gamma}^{c}=\frac{1}{2}C_{b\alpha}^{d}G_{dc}[\hat{\gamma}^{b},\hat{\gamma}^{c}]=4\Lambda_{\alpha}.
\end{gather*}
From the chain of equalities
\begin{gather*}
\sigma_{i}^{\alpha}(x,e_{H})\eta_{a}^{i}(x,e_{H})=\sigma_{A}^{\alpha}(x,e_{H})\eta_{a}^{A}(x,e_{H})-\sigma_{\beta}^{\alpha}(e_{H})\eta_{a}^{\beta}(x,e_{H})=\\
=\delta_{a}^{\alpha}-(-\delta_{\beta}^{\alpha})\eta_{a}^{\beta}(x,e_{H})=\eta_{a}^{\alpha}(x,e_{H}),
\end{gather*}
we obtain for the spinor connection the required expression (\ref{sp_gamma_x}).
\end{proof}

Thus, the Dirac equation in the homogeneous space $M$ with an invariant
metric $g_{M}$ and the Dirac matrices of the form (\ref{my_sol_gamma})
takes the form
\begin{gather*}
\mathcal{D}_{M}(x,\partial_{x})\psi=m\psi,\\
\mathcal{D}_{M}(x,\partial_{x})=i\hbar\hat{\gamma}^{a}\left[\eta_{a}^{i}(x,e_{H})\partial_{x^{i}}+\Gamma_{a}+\eta_{a}^{\alpha}(x,e_{H})\Lambda_{\alpha}\right].
\end{gather*}
A set of matrices $\Lambda_{\alpha}$ determines a spinor representation
of the isotopy subgroup $H$ in the space $V_{\Psi}$.
\begin{lemma}\label{prop2}
The matrices $\Lambda_{\alpha}$ are  generators
of the isotropy subgroup $H$ representation on the space $V_{\Psi}$.
\end{lemma}
\begin{proof}
We prove that the matrices $\Lambda_{\alpha}$ satisfy the commutation
relations
\begin{equation}
[\Lambda_{\alpha},\Lambda_{\beta}]=C_{\alpha\beta}^{\gamma}\Lambda_{\gamma}.\label{genH4}
\end{equation}
The commutator of  $\Lambda_{\alpha}$ and $\Lambda_{\beta}$
can be written as
\begin{align}
[\Lambda_{\alpha},\Lambda_{\beta}] & =  -\frac{1}{4}C_{\beta b}^{d}[\Lambda_{\alpha},\hat{\gamma}^{b}\hat{\gamma}_{d}]=  -\frac{1}{4}C_{\beta b}^{d}\left([\Lambda_{\alpha},\hat{\gamma}^{b}]\hat{\gamma}_{d}+\hat{\gamma}^{b}[\Lambda_{\alpha},\hat{\gamma}_{d}]\right).\label{comm3}
\end{align}
Using (\ref{conds_AdH}), (\ref{sys_gamma2}) and (\ref{gamma_down}),
we find the commutator of $\Lambda_{\alpha}$ with $\hat{\gamma}^{a}$:
\begin{align}
[\Lambda_{\alpha},\hat{\gamma}^{a}] & = \frac{1}{4}C_{b\alpha}^{d}[\hat{\gamma}^{b}\hat{\gamma}_{d},\hat{\gamma}^{a}]=\frac{1}{2}C_{b\alpha}^{d}\left(\delta_{d}^{a}\hat{\gamma}^{b}-G^{ab}\hat{\gamma}_{d}\right)=\nonumber \\
& =\frac{1}{2}\left(G^{bd}C_{b\alpha}^{a}-G^{ab}C_{b\alpha}^{d}\right)\hat{\gamma}_{d}=
C_{b\alpha}^{a}\hat{\gamma}^{b}.\label{comm1}
\end{align}
Similarly, for the $\gamma$-matrices with lower indices we have
\begin{equation}
[\Lambda_{\alpha},\hat{\gamma}_{a}]=C_{\alpha a}^{b}\hat{\gamma}_{b}.\label{comm2}
\end{equation}
Substitution of (\ref{comm1})--(\ref{comm2}) in (\ref{comm3}) yields
\begin{equation}
[\Lambda_{\alpha},\Lambda_{\beta}]=\frac{1}{4}\left(C_{\beta e}^{d}C_{\beta b}^{e}-C_{\beta b}^{e}C_{\alpha e}^{d}\right)\hat{\gamma}^{b}\hat{\gamma}_{d}.\label{commLL}
\end{equation}
The expression inside  the parentheses can be written in the form
\begin{gather}
C_{\beta e}^{d}C_{\beta b}^{e}-C_{\beta b}^{e}C_{\alpha e}^{d}=\left[C_{\alpha b}^{A}C_{\beta A}^{d}+C_{b\beta}^{A}C_{\alpha A}^{d}+C_{\beta\alpha}^{A}C_{bA}^{d}\right]+C_{\alpha\beta}^{\gamma}C_{b\gamma}^{d}.\label{CC_CC}
\end{gather}
By the Jacobi identity for the structure constants, the expression
inside the square brackets is equal to zero. Substituting (\ref{CC_CC}) in (\ref{commLL}),
we obtain (\ref{genH4}).
\end{proof}

The Dirac operator $\mathcal{D}_{M}(x)$ is a differential operator of the first order with matrix coefficients acting on the functions from $C^{\infty}(M,V_{\Psi})$ (spinors). Let us associate this operator with the corresponding operator on the transformation group. The key theorem of our work is

\begin{theorem}\label{prop3}
The Dirac operator $\mathcal{D}_{M}(x,\partial_{x})$ on the homogeneous space $M$ with invariant metric $g_{M}$ is a projection of an operator
\begin{gather}
\mathcal{D}_{M}(x,\partial_{x})=\hat{\pi}_{*}\mathcal{D}_{G}(g,\partial_{g}),\nonumber \\
\mathcal{D}_{G}(g,\partial_{g})\equiv i\hbar\hat{\gamma}^{a}[\eta_{a}(g)+\Gamma_{a}]\in L(\hat{\mathcal{F}}_{\Psi}).\label{dirac02}
\end{gather}
\end{theorem}
\begin{proof}
Comparing the Dirac operator $\mathcal{D}_{M}(x,\partial_{x})$ with
invariant matrix differential operator of the first order (\ref{R1form})
on the homogeneous space $M$, we obtain
\begin{equation}
B^{a}=i\hbar\hat{\gamma}^{a},\quad B=i\hbar\hat{\gamma}^{a}\Gamma_{a}.\label{coefBB}
\end{equation}
Projection (\ref{dirac02}) is determined if the coefficients $B^{a}$ and $B$ of the form (\ref{coefBB}) satisfy equations (\ref{sysBa})-–(\ref{sysB}). From (\ref{comm1}) it follows that the commutator of $\Lambda_{\alpha}$ and $\hat{\gamma}^{a}$ satisfies the first condition in (\ref{sysBa}). In this case the condition (\ref{sysB}) is reduced to
\begin{equation}
[\Gamma,\Lambda_{\alpha}]=C_{a\alpha}^{\beta}\hat{\gamma}^{a}\Lambda_{\beta}.\label{condLs2}
\end{equation}
The commutator of $\Gamma$ and $\Lambda_{\alpha}$ can be presented
in terms of the commutator $[\Lambda_{\alpha},\Gamma_{a}]$ as
\begin{equation}
[\Lambda_{\alpha},\Gamma]=[\Lambda_{\alpha},\hat{\gamma}^{a}]\Gamma_{a}+\hat{\gamma}^{a}[\Lambda_{\alpha},\Gamma_{a}].\label{commLG}
\end{equation}
Using (\ref{sp_gamma_x}) and 
(\ref{comm1})--(\ref{comm2}),
we get
\begin{align}
[\Lambda_{\alpha},\Gamma_{a}] 
& =-\frac{1}{4}\Gamma_{ba}^{d}\left([\Lambda_{\alpha},\hat{\gamma}^{b}]\hat{\gamma}_{d}+\hat{\gamma}^{b}[\Lambda_{\alpha},\hat{\gamma}_{d}]\right)=\nonumber \\
& =-\frac{1}{4}\Gamma_{ba}^{d}\left(C_{\alpha d}^{c}\hat{\gamma}^{b}\hat{\gamma}_{c}-C_{\alpha c}^{b}\hat{\gamma}^{c}\hat{\gamma}_{d}\right).\label{commLGa}
\end{align}
Substituting (\ref{commLGa}) into (\ref{commLG}), we find
\begin{equation}
[\Lambda_{\alpha},\Gamma]=\frac{1}{4}\left(C_{\alpha a}^{e}\Gamma_{be}^{c}-C_{\alpha c}^{c}\Gamma_{ba}^{e}+C_{\alpha b}^{e}\Gamma_{ea}^{c}\right)\hat{\gamma}^{a}\hat{\gamma}^{b}\hat{\gamma}_{c}.\label{commLG3}
\end{equation}
From (\ref{gamma_P}) and the Jacobi identity for structure constants
of the Lie algebra $\mathfrak{g}$ it follows that
\begin{equation}
C_{\alpha e}^{c}\Gamma_{ba}^{e}=C_{\alpha a}^{e}\Gamma_{be}^{c}+C_{\alpha b}^{e}\Gamma_{ea}^{c}+C_{\alpha a}^{\beta}C_{\beta b}^{c}.\label{CG3}
\end{equation}
Substituting (\ref{CG3}) into (\ref{commLG3}), we obtain (\ref{condLs2}).
Thus, relations (\ref{sysBa}) and (\ref{sysB}) are satisfied and
the lemma \ref{prop1} holds. It follows that the Dirac
operator $\mathcal{D}_{M}(x,\partial_{x})$ can be obtained by projection
of the operator $B^{a}\eta_{a}+B$ where $B^{a}$ and $B$ are given
by (\ref{coefBB}), and we come to (\ref{dirac02}).
\end{proof}
From this statement we immediately obtain
\begin{corollary}\label{coroll1} 
The generators
\begin{equation}
\widetilde{X}(x)=\xi_{X}^{a}(x)\partial_{x^{a}}+\xi_{X}^{\alpha}(x,e_{H})\Lambda_{\alpha}\label{killX0}
\end{equation}
of representation of the Lie algebra $\mathfrak{g}$ in the space $V_{\Psi}$
are symmetry operators of the Dirac operator $\mathcal{D}_{M}$ on
the homogeneous space $M$, i.e.
\begin{gather*}
[\widetilde{X}(x),\mathcal{D}_{M}(x,\partial_{x})]=0,\\
[\widetilde{X}(x),\widetilde{X}(x)]=\widetilde{[X,Y]}(x),\quad X,Y\in\mathfrak{g}.
\end{gather*}
\end{corollary}

In view of the isomorphism $\hat{\mathcal{F}} \simeq C^{\infty} (M, V_{\Psi})$ and Theorem \ref{prop3}, the Dirac equation (\ref{diracM}) on $M$ is equivalent to the following system of equations on the transformation group $G$:
\begin{equation}
\begin{cases}
\mathcal{D}_{G}(g,\partial_{g})\psi(g)=m\psi(g),\\
\left(\eta_{\alpha}+\Lambda_{\alpha}\right)\psi(g)=0.
\end{cases}\label{sysG}
\end{equation}

The function $\varphi (x) = \psi(s(x))=\psi(x, e_{H})$ provides  a solution of the original Dirac equation (\ref{diracM}) on the homogeneous space $M$.

\section{Noncommutative integration\label{sec:non}}

We will look for a solution to the system (\ref{sysG}) as a set of functions
\begin{equation}
\psi_{\sigma}(g)=\int_{Q}\psi_{\sigma}(q') {D}_{qq'}^{\lambda}\left(g^{-1}\right)d\mu(q'),\quad\sigma=(q,\lambda),\label{sub001}
\end{equation}
where the function $\psi_{\sigma}(q')$ is a spinor, each component of which belongs to the function space $L(Q, \mathfrak {h}, \lambda)$ with respect to the variable $q'$, and $ {D}_{qq'}^{\lambda}(g^{-1})$ is introduced by \eqref{funkD}.

Using (\ref{tD}) we can then reduce the system (\ref{sysG}) to the equations
\begin{gather}
\begin{cases}
\mathcal{D}_{\ell}(q',\partial_{q'},\lambda)\psi_{\sigma}(q')=m\psi_{\sigma}(q'),\\
\left(\ell_{\alpha}(q',\lambda)+\Lambda_{\alpha}\right)\psi_{\sigma}(q')=0,
\end{cases}\label{Dl}
\end{gather}
where
\begin{equation}\label{Dll3}
	\mathcal{D}_{\ell}(q',\partial_{q'},\lambda)=
i\hbar\hat{\gamma}^{a}[\ell_{a}(q',\lambda)+\Gamma_{a}].
\end{equation}
We call the operator $\mathcal{D}_{\ell}(q', \partial_{q'}, \lambda)$ in \eqref{Dll3} the Dirac operator in the $\lambda$-representation. The number of independent variables $q'$ in \eqref{Dl} is $\dim \,Q$. The set of functions $\psi_{\sigma} (x) = \psi_{\sigma}(s (x))$ gives us  a solution of the original Dirac equation (\ref{diracM}) on the homogeneous space $M$.

It follows from the equations
\begin{align*}
& \xi_{X}(g)\psi_{\sigma}(g)  =U(h)\widetilde{X}(x)\psi_{\sigma}(x)=-\ell_{X}(q,\lambda)\psi_{\sigma}(g)=U(h)\left\{ -\ell_{X}(q,\lambda)\psi_{\sigma}(x)\right\},
\end{align*}
that the solutions $\psi_{\sigma}(x)$ of the Dirac equation satisfy the system
\begin{align}
\left\{ \widetilde{X}(x)+\ell_{X}(q,\lambda)\right\} \psi_{\sigma}(x) & =0,\label{XLp}
\end{align}
where $\widetilde{X}$ is given by \eqref{xiX}.
The algebraic relations between  operators of the $\lambda$-representation should  correspond to the algebraic relations between the generators $\widetilde{X}(x)$ for compatibility of   the system \eqref{XLp}.
More precisely,  the corollary of the system (\ref{XLp}) is to be fulfilled for any homogeneous function $F$ of $\widetilde{X}(x)$:
\[
F(\widetilde{X}(x))\psi_{\sigma}(x)=F(-\ell_{X}(q,\partial_{q},\lambda))\psi_{\sigma}(x).
\]
This condition is obviously satisfied for the commutator of two operators
($X,Y\in\mathfrak{g}$),
\[
[\widetilde{X}(x),\widetilde{Y}(x)]\psi_{\sigma}(x)=-[\ell_{X}(q,\partial_{q},\lambda),\ell_{Y}(q,\partial_{q},\lambda)]\psi_{\sigma}(x),
\]
and for the Casimir functions, we have
\begin{align*}
 & K_{\mu}^{(s)}(i\hbar\widetilde{X})\psi_{\sigma}(x)=\left(U^{-1}(h)K_{\mu}^{(s)}(i\hbar\xi(g))U(h)\right)\psi_{\sigma}(x)=\\
 & =U^{-1}(h)K_{\mu}^{(s)}(i\hbar\xi(g))\psi_{\sigma}(g)=\kappa_{\mu}^{(s)}(\lambda)\psi_{\sigma}(x).
\end{align*}

The  homogeneous functions $\Gamma\in C^{\infty}(\mathfrak{g}^{*})$ provided that $\Gamma(X (x))\equiv 0$,  can exist on the dual space $\mathfrak{g}^{*}$ to  the   space $M$.
These functions are called \textit{identities on the homogeneous space} $M$. The number of functionally independent identities \textit{$i_{M}$ is called the \textit{index} of the homogeneous space $M$. }
In Ref. \cite{ShT} it was shown that  any homogeneous function $\Gamma\in C^{\infty}(\mathfrak{g}^{*})$
 satisfying the condition
\begin{equation}
\left.\Gamma(f)\right|_{f\in\mathfrak{h^{\perp}}}=0,\label{defG}
\end{equation}
is an identity. In the same Ref. \cite{ShT} it was shown that
 the functions $F_{\alpha}^{(s_{M})}(f)$,
where $s_{M}=\left (\dim \, \mathfrak{g^{\lambda}} - \mathrm{ind} \, \mathfrak{g}\right)/2$,
are identities, and all other identities are Casimir functions
 $ {K}_{\mu}^{(s_{M})}(f)$ such that
 \[
\left. {K}_{\mu}^{(s_{M})}(f)\right|_{f\in\mathfrak{h^{\perp}}}\equiv0.
\]

A description of identities for the  generators  $\widetilde{X}(x)$ of the form \eqref{xiX} is given by the following Lemma which is essential result of our work:
\begin{lemma}\label{Flemma}
The identities for the operators  $\widetilde{X}(x)$ are
generated by the functions $F_{\alpha}^{(s_{M})}(f)$ and $ {K}_{\mu}^{(s_{M})}(f)$.
\end{lemma}
\begin{proof}
Suppose that a homogeneous function $\Gamma'\in C^{\infty}(\mathfrak{g}^{*})$ is an  identity for the generators  $\widetilde X(x)$, i.e., $\Gamma' (\widetilde{X}) \equiv 0$. Then the symbol 
\[
  \Gamma'(x, p) = \Gamma' (\widetilde{X}(x, p)) \in C^{\infty}(\mathfrak{g}^{*}, V)
\]
of the operator $\Gamma'(\widetilde{X})$ also equals zero for all $(x, p)$, and 
\[
  \widetilde{X}(x, p) = X^{a}(x) p_{a} + \xi_{X}^{\alpha}(x, e_{H}) \Lambda_{\alpha},
\] 
where the constants $p_{a}$ are  coordinates of the covector $f^{x} = p_{a}dx^{a} \in T_{x}^{*} M$. At a given point $(x_{0}, p^{0})$, we have
\begin{gather*}
	\widetilde{X}_{a}(x_{0},p^{0})=p_{a}^{0},\quad\widetilde{X}_{\alpha}(x_{0},p^{0})=\Lambda_{\alpha},\\ f^{x_{0}}=p_{a}^{0}dx^{0}\in T_{x_{0}}^{*}M\simeq\mathfrak{h}^{\perp}.
\end{gather*}
Expanding   $\Gamma'(\widetilde{X}(x, p))$ in terms of the basis $ {B}$ of matrices in the vector space $V$ and putting $x = x_{0}$, $p = f^{x_{0}}$, we get:
\begin{gather*}
	\Gamma'(\widetilde{X}(x_{0},p^{0}))=\Gamma'(p_{1}^{0},\dots p_{\dim M}^{0},\Lambda_{1},\dots,\Lambda_{\dim H}) = {B}^{\sigma}\Gamma_{\sigma}(p_{1}^{0},\dots p_{\dim M}^{0})\equiv 0.
\end{gather*}
 As a result, for each function $\Gamma_{\sigma}(f) $ we  come to equation (\ref{defG}).
The last one shows that the functions $\Gamma_{\sigma}(f)$ are identities on a homogeneous space, and the identities $\Gamma'(f)$ for the operators
$\widetilde X(x)$ have the following structure:
\begin{align*}
	\Gamma'(\widetilde{X}(x)) & = \Gamma'(X(x)+\xi^{\alpha}(x,e_{H})\Lambda_{\alpha})= \\
	                          & = {B}^{\sigma}\Gamma_{\sigma}(X(x))\equiv 0.
\end{align*}
From this one can see that the number of functionally independent identities between 
$\widetilde{X}(x)$ does not exceed the index $i_{M}$ of the homogeneous space, and the functions $\Gamma_{\sigma}(f)$  depend on identities on the homogeneous space.
\end{proof}

For the compatibility of the system (\ref{XLp}), we have to take into account the identities between the generators $\widetilde{X}(x)$,   $\Gamma'(\widetilde{X}) \equiv 0$; namely we impose the following conditions on the operators of the $\lambda$-representation:
\begin{equation}
\Gamma'(-\ell(q,\partial_{q},\lambda))\equiv0.\label{tl1}
\end{equation}
A class of orbits and corresponding parameters  $j$ should be restricted by (\ref{tl1}).

For instance, for the case $\widetilde{X}(x) = X(x)$, the condition (\ref{tl1}) is reduced to
\begin{gather}
	F_{\alpha}^{(s_{M})}(-i\hbar\ell(q,\partial_{q},\lambda))\equiv 0,\nonumber \\
	{K}_{\mu}^{(s_{M})}(-i\hbar\ell(q,\partial_{q},\lambda))=\kappa_{\mu}^{(s_{M})}(\lambda(j))=0.\label{tl2}
\end{gather}
The first condition in (\ref{tl2}) says that the $\lambda$-representation has to be constructed by the class of orbits $\mathcal{O}_{\lambda(j)}^{(s_{M})}$, and the second one imposes a restriction on the parameters  $j$. In Ref.  \cite{BrGSh}, a $\lambda$-representation satisfying (\ref{tl2}) is called a $\lambda$-representation \textit{corresponding to the homogeneous space} $M$.

Thus, condition (\ref{tl2}) is stronger than (\ref{tl1}). One of the important results of our work is the fact that when performing a noncommutative reduction of the Dirac equation, it is necessary to use the weaker condition (\ref{tl1}) for the correct application of the noncommutative integration method.

The second equation of the system (\ref{Dl}) can be written as
\begin{align}
\bigg{(}\alpha_{\alpha}^{\overline{a}}(q')\frac{\partial}{\partial q'^{\overline{a}}}+\frac{i}{\hbar}\xi_{\alpha}^{\overline{\alpha}}(q',e_{P})&\left(\lambda_{\overline{\alpha}}+i\hbar\beta_{\overline{\alpha}}\right)+\Lambda_{\alpha}\bigg{)}\psi_{\sigma}(q')=0.\label{sys3}
\end{align}

We look for a  solution of \eqref{sys3}  in the form
\[
\psi_{\sigma}(q')=R_{\sigma}(q')\psi_{\sigma}(v),
\]
where $R_{\sigma}(q')$ is a certain function, and $\psi_{\sigma}(v)$ is an arbitrary function of the characteristics $v =v(q')$ of the system (\ref{Dl}).
We carry out a one-to-one change of variables $q'= q'(v, w)$, where $w = w (q')$ are some coordinates additional to $v$. By $V$ and $W$ we denote domains of the variables $v$ and $w$, respectively. The measure $d\mu (q')$ in the new variables takes the form $d\mu (q') = \rho (v, w) d\mu(v) d\mu( w)$. Then the solution of the original Dirac equation can be represented as
\begin{align*}
 & \psi_{\sigma}(x)=\int_{V}\psi_{\sigma}(v) {D}_{qv}^{\lambda}\left(x\right)d\mu(v),\\
 &  {D}_{qv}^{\lambda}\left(x\right)=\int_{W}R_{\sigma}(v,w) {D}_{qq'(v,w)}^{\lambda}\left(g^{-1}\right)\rho (v,w)d\mu(w).
\end{align*}
Substituting the solution $\psi_{\sigma} (x)$ into the Dirac equation (\ref{Dl}), we obtain a linear
first-order differential equation with matrix coefficients for the function $\psi_{\sigma}(v)$  with  the number of  independent variables $v$ equal to $\dim V$.

\section{The metric that does not admit separation of variables in the Dirac equation\label{sec:nosh}}

Consider a four-dimensional homogeneous space $M$ with a transformation group $G$ whose Lie algebra $\mathfrak{g}$ is defined in some basis $\{e_{A}\}$ by nonzero commutation relations
\begin{align*}
 & [e_{1},e_{4}]=-e_{1},\quad[e_{1},e_{5}]=e_{2},\quad[e_{2},e_{3}]=e_{1},\\
 & [e_{2},e_{4}]=e_{2},\quad[e_{3},e_{4}]=-2e_{3},\quad[e_{3},e_{5}]=e_{4},\\
 & [e_{4},e_{5}]=-2e_{5}.
\end{align*}
The Lie algebra $\mathfrak{g}$ is a semidirect product of the two-dimensional commutative ideal $\mathbb{R}^{2} = \mathrm{span} \{e_{1}, e_{2}\}$ and the three-dimensional simple algebra $\mathfrak{sl}(2) = \mathrm{span} \{e_{3}, e_{4}$, $e_{5}\}$. We also take $\mathfrak{h} = \{e_{5}\} $ as the one-dimensional subalgebra.

Denote by $(x^{a},h^{\alpha})$ local coordinates on a trivialization domain $U$ of the group $G$  so that
\begin{gather}
g(x,h)=e^{h^{1}e_{5}}e^{x^{4}e_{4}}e^{x^{3}e_{3}}e^{x^{2}e_{2}}e^{x^{1}e_{1}},\label{can3}\\
x=(x^{1},x^{2},x^{3},x^{4})\in M. \nonumber
\end{gather}

The group $G$ is unimodular and $\Delta_{G}(g) = 1$.
A symmetric non-degenerate matrix
\[
(G_{ab})=\begin{pmatrix}0 & 0 & 0 & -c_{3}\\
0 & c_{4} & c_{3} & c_{2}\\
0 & c_{3} & 0 & 0\\
-c_{3} & c_{2} & 0 & c_{1}
\end{pmatrix},\quad c_{k}=\mathrm{const},\quad c_{3}\neq0,
\]
defines an invariant metric on the  space $M$,
\begin{align}
ds^{2} & =\frac{e^{x^{4}}}{c_{3}^{2}}\bigg{\{} 2c_{2}e^{2x^{4}}\left(dx^{1}-x^{3}dx^{2}\right)dx^{3}-c_{1}e^{x^{4}}\left(dx^{1}-x^{3}dx^{2}\right){}^{2}-c_{4}e^{3(x^{4})^{2}}d(x^{3})^{2}+\nonumber\\
 &+2c_{3}\left(dx^{3}+x^{3}dx^{4}\right)dx^{2}-2c_{3}dx^{1}dx^{4}\bigg{\}} .\label{ds2primer}
\end{align}
The metric (\ref{ds2primer}) has nonzero scalar curvature $R=6c_{1}$. The group 
generators in the canonical coordinates (\ref{can3}) have the form:
\begin{align*}
X_{1}(x) & =\partial_{x^{1}},\quad X_{2}(x)=\partial_{x^{2}},\quad X_{3}(x)=x^{2}\partial_{x^{1}}+\partial_{x^{3}},\\
X_{4}(x) & =-x^{1}\partial_{x^{1}}+x^{2}\partial_{x^{2}}-2x^{3}\partial_{x^{3}}+\partial_{x^{4}},\\
X_{5}(x) & =x^{1}\partial_{x^{2}}-(x^{3})^{2}\partial_{x^{3}}+x^{3}\partial_{x^{4}}.
\end{align*}

The vector fields $\xi_{A}(x, e_{H})$, in turn, are determined by the expressions
\begin{align*}
\xi_{a}(x,h) & =X_{a}(x),\quad\xi_{5}(x,h)=X_{5}(x)+e^{-2x^{4}}\partial_{h^{1}}.
\end{align*}

The right-invariant vector fields in the canonical coordinates (\ref{can3}) are
\begin{gather*}
\eta_{1}(x,h)=-\left(e^{-x^{4}}+h_{1}x^{3}e^{x^{4}}\right)\partial_{x^{1}}-h_{1}e^{x^{4}}\partial_{x^{2}},\quad
\eta_{2}(x,h)=e^{-x^{4}}(x^{3}\partial_{x^{1}}+\partial_{x^{2}}),\\
\eta_{3}(x,h)=-e^{-2x^{4}}\partial_{x^{3}}-h_{1}\partial_{x^{4}}+h_{1}^{2}\partial_{h^{1}},\quad
\eta_{4}(x,h)=-\partial_{x^{4}}+2h_{1}\partial_{h^{1}},\\
\eta_{5}(x,e_{H})=-\partial_{h^{1}}.
\end{gather*}
The gamma matrices $\hat{\gamma}^{a}$ can be presented  in terms of the standard Dirac gamma matrices $\gamma^{a}$ as follows:
\begin{align*}
 & \hat{\gamma}^{1}=\gamma^{1}+\gamma^{2},\quad\hat{\gamma}^{2}=-\frac{c_{2}}{c_{3}}\left(\gamma^{1}+\gamma^{2}\right)+\sqrt{-c_{4}}\gamma^{3},\\
 & \hat{\gamma}^{3}=-\frac{c_{3}}{\sqrt{-c_{4}}}\left(\gamma^{3}+i\gamma^{4}\right),\\
 & \hat{\gamma}^{4}=-\frac{c_{3}}{2}\left(\gamma^{1}-\gamma^{2}\right)-\frac{2c_{2}}{\sqrt{-c_{4}}}\left(\gamma^{3}+i\gamma^{4}\right)-\frac{c_{1}}{c_{3}}\left(\gamma^{1}+\gamma^{2}\right).
\end{align*}
The spin connection is independent of local coordinates and has the form
\begin{align*}
 & \Gamma(x)=\Gamma-\hat{\gamma}^{4}\Lambda,\quad\Lambda=\frac{1}{2c_{3}}\hat{\gamma}^{1}\hat{\gamma}^{3},\\
 & \Gamma=\frac{c_{1}}{4c_{3}^{2}}\hat{\gamma}^{1}\hat{\gamma}^{1}\hat{\gamma}^{3}+\frac{c_{2}}{4c_{3}^{2}}\hat{\gamma}^{1}\hat{\gamma}^{3}\hat{\gamma}^{4}+\frac{1}{c_{3}}\hat{\gamma}^{2}\hat{\gamma}^{3}\hat{\gamma}^{4}-\frac{c_{1}}{4c_{3}}\hat{\gamma}^{1}+\frac{3c_{3}}{4c_{3}}\hat{\gamma}^{3}-2\hat{\gamma}^{4}.
\end{align*}
The Dirac operator in local coordinates is
\begin{align*}
 & \mathcal{D}_{M}=-i\hbar\bigg{(}\left[\hat{\gamma}^{1}+x^{3}e^{2x^{4}}\hat{\gamma}^{2}\right]\partial_{y^{1}}+e^{x^{4}}\hat{\gamma}^{2}\partial_{x^{2}} -e^{-2x^{4}}\hat{\gamma}^{3}\partial_{x^{3}}+\hat{\gamma}^{4}\partial_{x^{4}}\bigg{)}+\Gamma-\hat{\gamma}^{4}\Lambda.
\end{align*}
The first-order symmetry operators are defined by (\ref{killX0}):
\[
\widetilde{X}_{a}(x)=X_{a}(x),\quad\widetilde{X}_{5}(x)=X_{5}(x)+e^{-2x^{4}}\Lambda.
\]

The metric (\ref{ds2primer}) generally does not admit the Yano vector field and the Yano-Killing tensor field, so the Dirac equation does not admit spin symmetry operators. As a result, the Dirac equation has  only two commuting symmetry operators $\{\widetilde {X}_{1}(x), \widetilde{X}_{2}(x)\}$ of the first order. However, the Dirac equation admits a third-order symmetry operator
\begin{align}
K(\widetilde{X})=& \widetilde{X}_{5}(x)\cdot\widetilde{X}_{1}(x)\cdot\widetilde{X}_{1}(x)+\widetilde{X}_{1}(x)\cdot\widetilde{X}_{2}(x)\cdot\widetilde{X}_{4}(x)-\widetilde{X}_{2}(x)\cdot\widetilde{X}_{2}(x)\cdot\widetilde{X}_{4}(x),\label{kas11}
\end{align}
where $X\cdot Y = (XY + YX)/2$ is the symmetrized product of the operators $X$ and $Y$. As a consequence, the metric (\ref{ds2primer}) does not admit separation of variables for  the Dirac equation. Note that the Klein-Gordon equation also admits only three commuting symmetry operators $\{X_{1}(x), X_{2}(x), K(X)\}$. One of them, $K(X)$, is the third-order operator, and the Klein-Gordon equation also does not admit separation of variables.

We now carry out a noncommutative reduction of the Dirac equation.

First, we describe  orbits of the coadjoint representation of the Lie group $G$. The Lie algebra $\mathfrak{g}$ admits the Casimir function $K(f)=f_{5}f_{1}^{2} + f_{2}f_{4} f_{1} -f_{2}^{ 2} f_{3}$. Expansion (\ref{rg}) takes the form
\begin{align*}
 & M_{0}=\left\{ f\in\mathfrak{g}^{*}\mid\neg\left(f_{1}=f_{2}=0\right)\right\},\quad
 \mathcal{O}_{\omega}^{0}=\left\{ f\in M_{0}\mid K_{1}(f)=\omega_{1}^{0}\right\},\\
 & \Omega^{0}=\mathbb{R}^{1},\quad\mathrm{dim}\mathcal{O}_{\lambda}^{0}=4,\\
 & M_{1}=\left\{ f\in\mathfrak{g}^{*}\mid f_{1}=f_{2}=0,\neg\left(f=0\right)\right\} ,\quad
 \mathcal{O}_{\omega^{1}}^{1}=\left\{ f\in M_{0}\mid K_{1}^{(1)}(f)=\omega_{1}^{1}\right\} ,\\
 & M_{2}=\mathcal{O}^{2}=\left\{ f=0\right\} ,\\
 & \mathfrak{g}^{*}=M_{0}\cup M_{1}\cup M_{2},\quad
 M_{0}=\bigcup_{\omega^{0}\in\Omega^{0}}\mathcal{O}_{\omega}^{0},\quad M_{1}=\bigcup_{\omega^{1}\in\Omega^{1}}\mathcal{O}_{\omega^{1}}^{1},
\end{align*}
where $K_{1}^{(1)}(f) = f_{4}^{2} + 4f_{3} f_{5}$ is the Casimir function of  $s = 1$ type. Each non-degenerate orbit from the class $\mathcal{O}_{\omega}^{0}$ passes through the covector $\lambda(j) = (1,0,0,0, j)$ and $K_{1} (\lambda(j)) = j$, where $j\in\mathbb{R}$.

The covector $\lambda(j)$  admits a real polarization  $\mathfrak{p} = \left\{e_{1}, e_{2}, e_{5}\right\}$, and the $\lambda$-representation corresponding to the class of orbits $\mathcal{O}_{\lambda (j)}^{0}$ is given by
\begin{align}
 & \ell_{1}(q,\partial_{q},\lambda)=\frac{i}{\hbar}q_{1},\quad\ell_{2}(q,\partial_{q},\lambda)=-\frac{i}{\hbar}q_{2},\nonumber \\
 & \ell_{3}(q,\partial_{q},\lambda)=q_{1}\partial_{q^{2}},\quad
   \ell_{4}(q,\partial_{q},\lambda)=q_{1}\partial_{q^{1}}-q_{2}\partial_{q^{2}},\nonumber\\
 & \ell_{5}(q,\partial_{q},\lambda)=q_{2}\partial_{q^{1}}+\frac{i}{\hbar}\frac{j}{q_{1}^{2}},\nonumber \\
 & K_{1}(-i\hbar\ell)=j.\label{lp5}
\end{align}
The operators $-i\hbar\ell_{X} (q, \partial_{q}, \lambda)$ are symmetric with respect to the measure $ d\mu (q) = dq_{1} dq_{2}$, $Q = \mathbb{R }^{2}$. Now solving  equations (\ref{tD}), we  get
\begin{align*}
 {D}_{qq'}^{\lambda}(g^{-1}) =&\exp\left(\frac{i}{\hbar}\left[q_{2}x^{2}-q_{1}x^{1}+\frac{j}{q_{2}'}\left(\frac{e^{x^{4}}}{q_{1}}-\frac{1}{q_{1}'}\right)\right]\right)\times\\
 & \times\delta\left(-h_{1}q_{2}'+q_{1}e^{-x^{4}},q_{1}'\right)\delta\left(-e^{x^{4}}\left(q_{1}x^{3}-q_{2}\right),q_{2}'\right),
\end{align*}
where $\delta(q,q')$ is the generalized Dirac delta function. The completeness and orthogonality conditions (\ref{Dful})-(\ref{Dort}) are satisfied for the measure
\[
  d\mu(\text{\ensuremath{\lambda) = (2\pi\hbar)^{-3} dj}},\quad j \in\mathbb{R}^{1}.
\] 
Orbits from the class $\mathcal{O}_{\lambda (j)}^{0}$ satisfy the integral condition (\ref{conjO}).

The homogeneous space $M$ has zero index, $i_{M} = 0$, and does not have identities that have to be taken into account in the method of  noncommutative integration. So, the $\lambda$-representation (\ref{lp5}) corresponds to the homogeneous space $M$.

Integrating the equation 
\[
  \left(\ell_5 (q',\partial_{q'},\lambda) + \Lambda\right)\psi_{\sigma}(q') = 0,\quad\sigma=(q_{1},q_{2},j)
\]
yields
\begin{equation}
\psi_{\sigma}(q')=\text{\ensuremath{\exp\left(\frac{ij}{\hbar q'_{1}q'_{2}}-\frac{q'_{1}}{q'_{2}}\Lambda\right)}\ensuremath{\psi_{\sigma}(q'_{2}).}}\label{ww3}
\end{equation}

Substituting \eqref{ww3} into the Dirac equation in the $\lambda$-representation (\ref{Dl}), we obtain the ordinary differential equation for the spinor $\psi_{\sigma}$,
\begin{align*}
 & i\hbar\frac{d}{dq_{2}'}\psi_{\sigma}(q'_{2})=M(q_{2}')\psi_{\sigma}(q_{2}'),\\
 & M(q_{2}')=\frac{1}{c_{1}q_{2}'}\hat{\gamma}^{4}\left\{ i\hbar\Gamma+q_{2}'\hat{\gamma}^{2}+\frac{j}{q_{2}'^{2}}\hat{\gamma}^{3}-m\right\}.
\end{align*}
Then we obtain the solution as
\begin{equation}
\psi_{\sigma}(q'_{2})=C_{\sigma}\exp\left(-\frac{i}{\hbar}\int_{1}^{q_{2}'}M(u)du\right),\label{sol1}
\end{equation}
where $C_{\sigma}$ is the normalization factor.

The solution to the Dirac equation in local coordinates can be obtained  by substituting (\ref{ww3}) into (\ref{sub001}) and integrating over $q'$ for $h_{1}= 0$:
\begin{align}
&\psi_{\sigma}(x)  =\exp\left(\frac{\frac{ij}{\hbar q_{1}}-q_{1}e^{-2x^{4}}\Lambda}{q_{2}-q_{1}x^{3}}\right)\exp\left(\frac{i}{\hbar}\left[q_{2}x^{2}-q_{1}x^{1}\right]\right)\psi_{\sigma}\left(e^{x^{4}}\left(q_{2}-q_{1}x^{3}\right)\right).\label{sol2}
\end{align}
The equations  (\ref{sol1}) and (\ref{sol2}) provide  a complete and orthogonal set of solutions to the Dirac equation on the homogeneous space $M$ with the metric (\ref{ds2primer}).

\section{The Dirac equation in  $AdS_{3}$ space\label{sec:ads3}}

Consider a three-dimensional de Sitter space $M$ as  a homogeneous space with the  de Sitter group of transformations $G=SO(1,3)$ and  the isotropy subgroup $H$ being  the Lorentz group $SO(1,2)$.
The de Sitter space $M$ has constant positive curvature and $M$ is topologically isomorphic to $R^{1}\times S^{2}$.

The group $SO(1,3)$ is defined as a rotation group of a 4-dimensional pseudo-Euclidean space with the metric $(G_{AB}) =\mathrm{diag} (1, -1, -1, -1)$. The Lie algebra $\mathfrak{g} = \mathfrak{so}(1,3)$ of the group $SO(1,3)$ in the basis $\left\{E_{AB}\mid A <B \right \}$ is defined by the commutation relations
\begin{align*}
   [E_{AB},E_{CD}]  = G_{AD}E_{BC}-G_{AC}E_{BD} +
   G_{BC}E_{AD}-G_{BD}E_{AC},
\end{align*}
where $A,B,C,D = 1,\dots , 4$. The basis $E_{AB}$ can be  represented as
\begin{gather*}
    E_{ab}=e_{ab},\ (a<b),\quad E_{a4}=e_{a}/\varepsilon,\\
    [e_{a},e_{b}]=\varepsilon^{2}e_{ab},\quad a,b=1,\dots,3.
\end{gather*}
Here the basis $e_{ab}$ forms the isotropy subalgebra $\mathfrak{h} = \mathfrak{so}(1,2)$, and $\varepsilon$ is a parameter defining the curvature of the de Sitter space. We introduce the canonical coordinates of the second kind on the Lie group $G$ so that
\begin{equation}
   g(t,x,y,h_{1},h_{2},h_{3})=
   e^{h_{3}e_{23}}e^{h_{2}e_{13}}e^{h_{1}e_{12}}e^{ye_{3}}e^{xe_{2}}e^{te_{1}},\label{can_so13}
\end{equation}
The group generators in the canonical coordinates can be written as
\begin{align*}
X_{1}(x) & =\partial_{t},\\
X_{2}(x) & =\sinh(\varepsilon t)\tan(\varepsilon x)\partial_{t}+\cosh(\varepsilon t)\partial_{x},\\
X_{3}(x) & =\sinh(\varepsilon t)\sec(\varepsilon x)\tan(\varepsilon y)\partial_{t}+\\
         & +\cosh(\varepsilon t)\sin(\varepsilon x)\tan(\varepsilon y)\partial_{x}+\\
         & +\cosh(\varepsilon t)\cos(\varepsilon x)\partial_{y},
\end{align*}
\begin{align*}
X_{12}(x) & =\varepsilon^{-1}\left(\cosh(\varepsilon t)\tan(\varepsilon x)\partial_{t}+\sinh(\varepsilon t)\partial_{x}\right),\\
X_{13}(x) & =\varepsilon^{-1}\big{(}\cosh(\varepsilon t)\sec(\varepsilon x)\tan(\varepsilon y)\partial_{t}+\\
         &+\sinh(\varepsilon t)\sin(\varepsilon x)\tan(\varepsilon y)\partial_{x}+\\
         &+\sinh(\varepsilon t)\cos(\varepsilon x)\partial_{y}\big{)},\\
X_{23}(x) & =\varepsilon^{-1}\left(\cos(\varepsilon x)\tan(\varepsilon y)\partial_{x}-\sin(\varepsilon x)\partial_{y}\right).
\end{align*}

The vector fields $\xi(x, e_{H})$, in turn, are determined by
\begin{align*}
\xi_{1}(x,e_{H}) & =X_{1}(x),\\
\xi_{2}(x,e_{H}) & =X_{2}(x)+\epsilon\sinh(t\epsilon)\sec(x\epsilon)\partial_{h^{1}},\\
\xi_{3}(x,e_{H}) & =X_{3}(x)+\epsilon\sinh(t\epsilon)\tan(x\epsilon)\tan(y\epsilon)\partial_{h^{1}}\\
 & +\epsilon\sinh(t\epsilon)\sec(y\epsilon)\partial_{h^{2}}+\\
 & +\epsilon\cosh(t\epsilon)\sin(x\epsilon)\sec(y\epsilon)\partial_{h^{3}},
\end{align*}
\begin{align*}
\xi_{12}(x,e_{H}) & =X_{12}(x)+\cosh(t\epsilon)\sec(x\epsilon)\partial_{h^{1}},\\
\xi_{13}(x,e_{H}) & =X_{13}(x)+\cosh(t\epsilon)\tan(x\epsilon)\tan(y\epsilon)\partial_{h^{1}}\\
 & +\cosh(t\epsilon)\sec(y\epsilon)\partial_{h^{2}}+\\
 & +\sinh(t\epsilon)\sin(x\epsilon)\sec(y\epsilon)\partial_{h^{3}},\\
\xi_{14}(x,e_{H}) & =X_{14}(x)+\cos(x\epsilon)\sec(y\epsilon)\partial_{h^{3}}.
\end{align*}

The right-invariant vector fields in the canonical coordinates (\ref{can_so13}) are
\begin{align*}
\eta_{1}(x,e_{H}) = &-\sec(\varepsilon x)\sec(\varepsilon y)\partial_{x}-\\
   &-\varepsilon\tan(\varepsilon x)\sec(\varepsilon y)\partial_{h^{1}}-
   \varepsilon\tan(\varepsilon y)\partial_{h^{2}},\\
\eta_{2}(x,e_{H}) = &-\sec(\varepsilon y)\partial_{y}-\varepsilon\tan(\varepsilon y)\partial_{h^{3}},\\
\eta_{3}(x,e_{H}) = &-\partial_{z},\quad
\eta_{12}(x,e_{H})= -\partial_{h^{1}},\\
\eta_{13}(x,e_{H})= &-\partial_{h^{2}},\quad
\eta_{23}(x,e_{H})= -\partial_{h^{3}}.
\end{align*}
The $2$--form $(G_{ab}) = \mathrm{diag} (1, -1, -1)$ defines an invariant metric on the space $M$  in local coordinates as
\begin{align*}
& ds^{2} =  g_{ij}(x)dx^{i}dx^{j} =
\rho(x,y)^{2}\left\{ dt^{2}-\frac{dx^{2}}{\cos^{2}(\varepsilon x)}-\frac{dy^{2}}{\rho(x,y)^{2}}\right\} ,\\
&\rho(x,y)=  \cos(\varepsilon x)\cos(\varepsilon y).
\end{align*}
Whence the scalar curvature of the space $M$ reads $R=6\varepsilon^{2}$.

The gamma matrices $\hat{\gamma}^{a}$ in the $(2 + 1)$-dimensional space in terms of the Pauli matrices $\left\{\sigma_{x}, \sigma_{y}, \sigma_{z}\right\}$ are
\[
  \hat{\gamma}^{1}=\sigma_{z},\quad\hat{\gamma}^{2}=is\sigma_{x},\quad\hat{\gamma}^{3}=i\sigma_{y}.
\]
Here the parameter $s=\mathrm{Tr}\hat{\gamma}^{1}\hat{\gamma}^{2}\hat{\gamma}^{3} = \pm 1$ is called a pseudospin. The matrices
\[
\Lambda_{12}=-\frac{is}{2}\hat{\gamma}^{3},\quad\Lambda_{13}=\frac{is}{2}\hat{\gamma}^{2},\quad\Lambda_{23}=-\frac{is}{2}\hat{\gamma}^{1}
\]
realize a spinor representation of the isotropy subalgebra in a space of two-dimensional spinors. In our case $\Gamma_{a} = 0$ and the spin connection takes the simple form
\begin{align*}
  \Gamma(x) = &\hat{\gamma}^{a}\eta_{a}^{\alpha}(x,e_{H})\Lambda_{\alpha}=
            \frac{\varepsilon}{2}\hat{\gamma}^{2}\tan(\varepsilon x)\sec(\varepsilon y)+\varepsilon\hat{\gamma}^{3}\tan(\varepsilon y).
\end{align*}
Then the Dirac operator in local coordinates reads
\begin{align*}
\mathcal{D}_{M} = -i\sec(\varepsilon x)\sec(\varepsilon y)\hat{\gamma}^{1}\partial_{t} -i\sec(\varepsilon y)\hat{\gamma}^{2}\left(\partial_{x}-\frac{\varepsilon}{2}\tan(\varepsilon x)\right) -i\hat{\gamma}^{3}\left(\partial_{y}-\varepsilon\tan(\varepsilon y)\right).
\end{align*}
The first-order symmetry operators, as defined by (\ref{killX0}),  are given by
\begin{align*}
\widetilde{X}_{1}(x) & =X_{1}(x),\\
\widetilde{X}_{2}(x) & =X_{2}(x)-\frac{is\varepsilon}{2}\sinh(\varepsilon t)\sec(\varepsilon x)\hat{\gamma}^{3},\\
\widetilde{X}_{3}(x)& =X_{3}(x)-\frac{is\varepsilon}{2}\left(\cosh(\varepsilon t)\sin(\varepsilon x)\sec(\varepsilon y)\hat{\gamma}^{1}-\right.\\
 & -\left.\sinh(\varepsilon t)\sec(\varepsilon y)\hat{\gamma}^{2}-\right.\\
 &  -\left.\sinh(\varepsilon t)\tan(\varepsilon x)\tan(\varepsilon y)\hat{\gamma}^{3}\right),\\
\widetilde{X}_{12}(x) & =X_{12}(x)-\frac{is\varepsilon}{2}\cosh(\varepsilon t)\sec(\varepsilon x)\hat{\gamma}^{3},\\
\widetilde{X}_{13}(x) & =X_{13}(x)-\frac{is}{2}\left(\sinh(\varepsilon t)\sin(\varepsilon x)\sec(\varepsilon y)\hat{\gamma}^{1}-\right.\\
 & -\left.\cosh(\varepsilon t)\sec(\varepsilon y)\hat{\gamma}^{2}+\right.\\
  & +\left.\cosh(\varepsilon t)\tan(\varepsilon x)\tan(\varepsilon y)\hat{\gamma}^{3}\right),\\
\widetilde{X}_{14}(x) & =X_{14}(x)-\frac{is\varepsilon}{2}\cos(\varepsilon x)\sec(\varepsilon y)\hat{\gamma}^{1}.
\end{align*}
Our aim is to construct  a complete set of exact solutions of the Dirac equation corresponding to the operator $\mathcal {D}_{M}$ using the noncommutative integration method.

The Lie algebra $\mathfrak{g}$ admits  the following two Casimir functions:
\[
K_{1}(f)=p^{a}p_{a}-\frac{\varepsilon^{2}}{2}l^{ab}l_{ab},\quad K_{2}(f)=\frac{\varepsilon}{8}e^{ABCD}L_{AB}L_{CD},
\]
where $e^{ABCD}$ is the Levi-Civita symbol ($e^{1234} =1$); $(p_{a}, l_{ab})$ are coordinates of the covector with respect to the basis $\left\{e^{a}, e^{ab} \right\}$, i.e.  $f = p_{a} e^{ a} + l_{ab} e^{ab}$, and $L_{AB}$ are  the  coordinates with respect to the basis $\left\{E^{AB}\right\}$,   $f = L_{AB} E^{AB}$. The raising and lowering indices are performed using the matrix $G_{AB}$. The expansion (\ref{rg}) in our case takes  the form:
\begin{align*}
 & M_{0}=\left\{ f\in\mathfrak{g}^{*}\mid\neg\left(f=0\right)\right\},\\
 & \mathcal{O}_{\omega}^{0}=\left\{ f\in M_{0}\mid K_{1}(f)=\omega_{1}^{0},\quad K_{2}(f)=\omega_{2}^{0}\right\},\\
 & \Omega^{0}=\mathbb{R}^{2},\quad\mathrm{dim}\mathcal{O}_{\lambda}^{0}=4,\\
 & M_{1}=\left\{ \emptyset\right\} ,\quad\mathcal{O}^{1}=\left\{ \emptyset\right\} ,\quad M_{2}=\mathcal{O}^{2}=\left\{ f=0\right\} ,\\
 & \mathfrak{g}^{*}=M_{0}\cup M_{2},\quad M_{0}=\bigcup_{\omega^{0}\in\Omega^{0}}\mathcal{O}_{\omega}^{0}.
\end{align*}
Each non-degenerate orbit from the class $\mathcal{O}_{\omega}^{0}$ passes through the covector $\lambda(j)=(j_{1}, 0,0,0,0, j_{2})$  characterized by two real parameters $j = (j_{1}, j_{2})\in \mathbb{R}^{2}$,  and
\begin{align*}
& K_{1}(\lambda(j)) =\omega_{1}^{0}(j)=j_{1}^{2}-\varepsilon^{2}j_{2}^{2},\quad K_{2}(\lambda(j)) =\omega_{2}^{0}(j)=j_{1}j_{2},\\
& \mathrm{det}\left\Vert \frac{\partial\omega_{\mu}^{0}(j)}{\partial j_{\nu}}\right\Vert =2\left(j_{1}^{2}+\varepsilon^{2}j_{2}^{2}\right).
\end{align*}
If  $K_{2}(f)=0$ for $f\in\mathfrak{p}^{\perp} = \text{\ensuremath{\left\{f_{a}e^{a}\mid f\in\mathfrak{g}^{*}\right \}}}$, then the Casimir operator $K_{2}(X)$ is an identity on the homogeneous space $M$:
\[
  K_{2}(X)=X_{3}\cdot X_{12}-X_{2}\cdot X_{13}+X_{1}\cdot X_{23}=0.
\]
Note that $K_{2}(\widetilde{X})$ is proportional to the Dirac operator:
\begin{align*}
	 K_{2}(\widetilde{X}) &  = \widetilde{X}_{3}\cdot\widetilde{X}_{12}-\widetilde{X}_{2}\cdot\widetilde{X}_{13}+\widetilde{X}_{1}\cdot\widetilde{X}_{23}=\frac{s}{2}\mathcal{D}_{M}.
\end{align*}
The covector $\lambda(j)$ admits the real polarization
\[
  \mathfrak{p} = \left\{e_{1}, e_{2} + \varepsilon e_{12}, e_ {3} + \varepsilon e_{13} , e_{23} \right \}.
\]
The corresponding $\lambda$-representation for the class of orbits
$\mathcal{O}_{\lambda (j)}^{0}$ is represented in Appendix A (see Eqs. (\ref{s013l})). The Casimir operators in the $\lambda$-representation are
\begin{align*}
   & K_{1}(-i\hbar\ell) = \kappa_{2}(\lambda) = j_{1}^{2}-\varepsilon^{2}j_{2}^{2}+(\varepsilon\hbar)^{2},\\ 
   & K_{2}(-i\hbar\ell)=\kappa_{2}(\lambda)=j_{1}j_{2}.
\end{align*}
The equation $\left(\ell_{\alpha}(q', \partial_{q'}, \lambda) + \Lambda_{\alpha}\right) c_{\lambda}(q ') = 0$ provided that  $j_{2} = s/2$ has a nonzero solution
\begin{align*}
c_{\lambda}(q) & = \frac{(\cos(\varepsilon q_{1})\cos(\varepsilon q_{2}))^{-\frac{3}{2}-\frac{ij_{1}}{\hbar\varepsilon}}}{\cos(\varepsilon q_{1})\cos(\varepsilon q_{2})+1}\times\\
&\times\bigg{(}\cos\left(\frac{\varepsilon(q_{1}+q_{2})}{2}\right)+is\cos\left(\frac{\varepsilon(q_{1}-q_{2})}{2}\right)\bigg{)}\begin{pmatrix}\cos(\varepsilon q_{1})\cos(\varepsilon q_{2})+1\\
is\sin(\varepsilon q_{1})-\cos(\varepsilon q_{1})\sin(\varepsilon q_{2})
\end{pmatrix}.
\end{align*}

The Dirac equation in the $\lambda$-representation,
\[
\left\{ i\hbar\hat{\gamma}^{a}\ell_{a}(q',\partial_{q'},\lambda)-m\right\} c_{\lambda}(q')=0,
\]
is reduced to the algebraic equation $j_{1} + m = 0$, then we have  $j_{1} =-m$ and $j_{2} = s/2$. That is, the eigenvalue of the Casimir operator $K_{1} (i\hbar X)$ is determined by  the particle mass $m$, and the eigenvalue of the second Casimir operator, $K_{2}(i \hbar X)$, depends on the parameter $s$:
\[
\kappa_{1}(\lambda)=m^{2}-\frac{1}{4}\varepsilon^{2}+(\varepsilon\hbar)^{2},\quad\kappa_{2}(\lambda)=-\frac{1}{2}ms.
\]
The solution of the original Dirac equation in our case reads
\begin{gather}
\psi_{\sigma}(x)=e^{-t\ell_{1}(q,\partial_{q},\lambda)}e^{-x\ell_{2}(q,\partial_{q},\lambda)}e^{-y\ell_{3}(q,\partial_{q}.\lambda)}c_{\lambda}(q),\quad \sigma=(q_{1},q_{2})\label{solads3}
\end{gather}
Here,  the exponentials of  operators of the $\lambda$-representa\-tion for the fixed $j_{1} = - m$ and $j_{2} = s/2$  act on a function according to Appendix A, Eqs. (\ref{expl}). From here one can  see that the solution  \eqref{solads3} depends on  two quantum numbers $q_{1}$ and $q_{2}$, which are not eigenvalues  for symmetry operators. The explicit form of the solution (\ref{solads3}) is cumbersome, but it is expressed in terms of elementary functions.

\section{Conclusion} \label{sec:Conclusion-remarks}

In this paper, we have explored the Dirac equation with an invariant metric on a homogeneous space $M$ of arbitrary dimension and developed the noncommutative integration method for this equation based on the ideas of symmetry analysis and the Lie group theory.

The Dirac equation and its symmetry are convenient to study in terms of algebraic structures associated with homogeneous spaces, and the theory of Lie group representations can be effectively applied for constructing exact solutions.

Using a special choice of the local frame and right-invariant vector fields on the Lie group  of  transformations $G$, we have obtained the spin connection (\ref{sp_gamma_x}). The Dirac equation is shown to be equivalent to a system of linear differential equations with constant matrix coefficients on the Lie group $G$ given by (\ref{sysG}) that is the starting point for noncommutative integration.

In Refs. \cite{fed01,sh02,var03,var04,kls01,kls02} an early version of the NCIM  was used to construct exact solutions to the Dirac equation in four-dimensional space-time where the $\lambda$-representa\-tion was constructed directly by definition \eqref{defL0}, and the domain of the variables $q$ was not associated  with a Lagrangian submanifold to the K-orbit. The desired solutions were constructed by means of joint integration of the system (\ref{XLp}) in local coordinates together with the original Dirac equation \eqref{diracM}.

Differently to the above early method, the main idea here is the noncommutative reduction of the corresponding system of equations on the Lie group $G$ and the connection between the solutions of this system and the original Dirac equation.

The noncommutative reduction is defined here using a special irreducible $\lambda$-representation of the Lie algebra $\mathfrak{g}$ of the Lie group $G$, which we introduce using the orbit method \cite{Kirr}. The key point of the method developed is based on the fact that there exist the identities connecting symmetry operators on a homogeneous space. For the Dirac equation, as follows from the lemma \ref{Flemma}, the number of identities is either less than for the Klein--Gordon equation or they are completely absent. For the Klein--Gordon equation, the number of identities is determined by the index of the homogeneous space \cite {ShBr11}.

The problem of describing identities for the symmetry operators of the Dirac equation on the homogeneous space is constructively solved for the first time. The reduced system (\ref{Dl}) in the $\lambda$-representation depending on a smaller number of independent variables $q'$ is obtained. What is remarkable is the fact that the solutions obtained for the Dirac equation on a homogeneous space are closely related to the two principal bundles of the transformation group. The local coordinates $x$ on the space $M$ are determined by means of the principal $H$-bundle of the group $G$ with the isotropy subgroup $H$ of $M$, and the set of quantum numbers $q$ is connected with the principal $P$-bundle of the same group $G$ and the subgroup $P$.

The NCIM developed in the paper for the Dirac equation is illustrated by two non-trivial examples. In one of them, described in Section \ref{sec:nosh},  we have found by using the NCIM  a complete set of solutions to the Dirac equation (the MCIM-solutions) in the case when the metric does not admit separation of variables neither in the Klein-Gordon equation nor in the Dirac equation. The solutions obtained are eigenfunctions of the symmetry operator of the third-order  (\ref{kas11}) and are parameterized by three parameters $(q_{1}, q_{2}, j)$.

The second example is the three-dimensional de Sitter space $\mathrm{AdS_{3}}$ with the transformation group $SO(1,3)$ (Section \ref{sec:ads3}). In this case, the Casimir operator $K_{2}(- i\hbar \widetilde{X})$ is proportional to the Dirac operator: 
\[
  \mathcal{D}_{M}=2sK (-i \hbar \widetilde {X}),
\] 
and the spectrum of the Casimir operator $ -K_{1} (- i \hbar \widetilde {X})$ gives the mass $m$ of the spinor field.

Noteworthy, the function $K_{2}(f)$ is an identity in the homogeneous space $M$, but when substituting the extended operators $-i\hbar\widetilde{X}$, it is no longer an identity, and this leads to the original Dirac operator. If we consider the Klein--Gordon equation in $\mathrm{AdS_{3}}$, the operator $-K_{1}(- i\hbar\widetilde{X})$ is proportional to the operator of the equation, and the second operator $K_{2}(- i\hbar X)$ is identically zero that  corresponds to the case $s=0$. So the noncommutative integration of the Dirac equation is different from the noncommutative integration of the Klein-Gordon equation.
The complete set of exact solutions (\ref{solads3}) of the Dirac equation found by the NCIM is parameterized by two real parameters $(q_{1}, q_{2})$ and is expressed by means of elementary functions, while the separation of variables gives the basis solutions to the Dirac equation in terms of special functions.

The parameters $q$ of solutions (\ref{sol2}) and (\ref{solads3}) obtained by the NCIM are in general not eigenvalues of an operator, a fact that crucially distinguishes them from solutions obtained by  separation of variables. Nevertheless, the NCIM-solutions can be effectively applied in order to study quantum effects in homogeneous spaces (see, e.g., \cite{ShBr11,ShBr14}).

The NCIM-solutions of the Dirac equation  may have a wide range of applications in the theory of fermion fields \cite{Hck,Toms}, quantum cosmology \cite{no2,BrM} and other problems of  field theory.  The NCIM  can be  applied also  to the Dirac-type  equation  for  theoretical models in  the condensed matter (graphene, topological insulators, etc.) \cite{vaf,klm}. Note that the technique proposed in the article can be easily generalized to the case of spaces having  new spatial dimensions much larger than the weak scale, as large as a millimeter for the case of two extra dimensions \cite{tev}.

Finally, we note that the NCIM reveals new aspects, both related to the symmetry of the Dirac equation and its integrability, and to study the properties of new solutions constructed. One of the problems is to
 find out the meaning of the parameters $q$ entering into the NCIM-solutions which, in the general case, do not have to be eigenvalues of operators representing observables. One can notice some similarity of the NCIM-solutions  with   well-studied coherent states \cite{per}. In particular, the action of the group on the set of $Q$ data of quantum numbers is defined, that can be found in the theory of coherent states \cite{mm, per, per2, Bagrov2}. However, the analysis of the parameters is the subject of special research.
 
 \section*{Acknowledgements}
 
 Breev and Shapovalov were partially supported by Tomsk State University under the International Competitiveness Improvement Program; Breev was partially supported by the Russian Foundation for Basic Research (RFBR) under the project No. 18-02- 00149; Shapovalov was partially supported by Tomsk Polytechnic University under the International Competitiveness Improvement Program and by RFBR and Tomsk region according to the research project No. 19-41-700004.

\section*{Appendix A. $\lambda$-representation of Lie algebra $\mathfrak{so}(1,3)$}

The $\lambda$ -representation for the class of orbits $\mathcal {O}_{\lambda (j)}^{0}$
is written as
\begin{align*}
\ell_{1}(q,\partial_{q},\lambda)  =&\sin(\varepsilon q_{1})\cos(\varepsilon q_{2})\partial_{q^{1}}+ 
\sec(\varepsilon q_{1})\sin(\varepsilon q_{2})\partial_{q^{2}}+\\ 
+&\left(\varepsilon+\frac{i}{\hbar}j_{1}\right)\cos(\varepsilon q_{1})\cos(\varepsilon q_{2})
+\frac{i}{\hbar}\varepsilon j_{2}\tan(\varepsilon q_{1})\sin(\varepsilon q_{2}),\\ 
\ell_{2}(q,\partial_{q},\lambda)  =&\cos(\varepsilon q_{2})\partial_{q^{1}}+\tan(\varepsilon q_{1})\sin(\varepsilon q_{2})\partial_{q^{2}} 
-\frac{i}{\hbar}\varepsilon j_{2}\sec(\varepsilon q_{1})\sin(\varepsilon q_{2}),\\ 
\ell_{3}(q,\partial_{q},\lambda)  =&\partial_{q^{2}},\\ 
\ell_{12}(q,\partial_{q},\lambda)  =&\varepsilon^{-1}\left[-\cos(\varepsilon q_{1})\partial_{q^{1}}+\left(\varepsilon+\frac{i}{\hbar}j_{1}\right)\sin(\varepsilon q_{1})\right],
\end{align*}
\begin{align*}
\ell_{13}(q,\partial_{q},\lambda) = &\varepsilon^{-1}\big{[}
\sin(\varepsilon q_{1})\sin(\varepsilon q_{2})\partial_{q^{1}} 
-\sec(\varepsilon q_{1})\cos(\varepsilon q_{2})\partial_{q^{2}}+\\ 
+&\left(\varepsilon+\frac{i}{\hbar}j_{1}\right)\cos(\varepsilon q_{1})
  \sin(\varepsilon q_{2})
+\frac{i}{\hbar}\varepsilon j_{2}\tan(\varepsilon q_{1})\cos(\varepsilon q_{2})\big{]},
\end{align*}
\begin{align}
	\ell_{23}(q,\partial_{q},\lambda)  =&\varepsilon^{-1}\left[\sin(\varepsilon q_{2})\partial_{q^{1}}-\tan(\varepsilon q_{1})\cos(\varepsilon q_{2})\partial_{q^{2}}
 +\frac{i}{\hbar}\varepsilon j_{2}\sec(\varepsilon q_{1})\cos(\varepsilon q_{2})\right]. \label{s013l}
\end{align}
The exponentials of operators $\ell_{1}(q,\partial_{q},\lambda)$, $\ell_{2}(q,\partial_{q},\lambda)$ and $\ell_{3}(q,\partial_{q},\lambda)$ for the fixed 
$j_{1} = - m$ and $j_{2} = s/2$ act on a function as follows:
\[
e^{-y\ell_{3}(q,\partial_{q},\lambda)}f(q)=f(q_{1},q_{2}-y),
\]
\begin{align*}
 & e^{-x\ell_{2}(q,\partial_{q},\lambda)}f(q)=e^{\frac{is}{2}\mathrm{\mathrm{arccot}} \left(\frac{\cot\varepsilon q_{2}}{\sin\varepsilon q_{1}}\right)}\Phi_{2}\left[\cos\varepsilon q_{1}\sin\varepsilon q_{2},\arctan\left(\frac{\tan\varepsilon q_{1}}{\cos\varepsilon q_{2}}\right)-\varepsilon x\right],\\
 & \Phi_{2}(a,b)=e^{-\frac{1}{2}is\,\arctan(a\tan b)}f\bigg{[}\varepsilon^{-1}\mathrm{arctan}
 \left(\frac{\sqrt{1-a^{2}}\tan b}{\sqrt{a^{2}\tan^{2}b+1}}\right),
 \varepsilon^{-1}\mathrm{arcsec}\left(\sqrt{\frac{a^{2}\tan^{2}b+1}{1-a^{2}}}\right)\bigg{]},
\end{align*}
\begin{align}\nonumber
 & e^{-t\ell_{1}(q,\partial_{q},\lambda)}f(q)=\frac{(\sin\varepsilon q_{1})^{im/\varepsilon}}{\sqrt{\sin\varepsilon q_{1}\left(\sin\varepsilon q_{1}\cos\varepsilon q_{2}-is\sin\varepsilon q_{2}\right)}}\times\\ \nonumber
 & \times\Phi_{3}\left[\sin\varepsilon q_{2}\cot\varepsilon q_{1},\frac{\cos\varepsilon q_{1}\cos\varepsilon q_{2}\sinh\varepsilon t+\cosh\varepsilon t}{\sqrt{1-\left(\cos\varepsilon q_{1}\cos\varepsilon q_{2}\right)^{2}}}\right],\\ \nonumber
 & \Phi_{3}(a,b) =\left(b\sqrt{a^{2}+1}\right)^{im/\varepsilon}\left(\frac{\sqrt{b^{2}-1}-is\,ab}{b^{2}\sqrt{a^{2}+1}\sqrt{\left(a^{2}+1\right)b^{2}-1}}\right)^{1/2}\times\\ 
  &f\bigg{[}\varepsilon^{-1}\mathrm{arcsec}\left(\frac{\sqrt{a^{2}+1}b}{\sqrt{\left(a^{2}+1\right)b^{2}-1}}\right),\varepsilon^{-1}\mathrm{arccos}\left(\frac{\sqrt{a^{2}+1}\sqrt{b^{2}-1}}{\sqrt{\left(a^{2}+1\right)b^{2}-1}}\right)\bigg{]}.\label{expl}
\end{align}

\end{document}